\documentclass[journal]{IEEEtran}

\IEEEoverridecommandlockouts
\usepackage{times,enumerate} 
\usepackage{color}
\usepackage{graphicx}
\usepackage{setspace}
\usepackage{bbm}
\usepackage{mathdots,mathrsfs}
\usepackage{amssymb,latexsym,amsfonts,amsmath,cite,comment,relsize}
\usepackage{stmaryrd}
\usepackage[small,bf]{caption}
\usepackage[dvips]{psfrag}
\usepackage{setspace}
\usepackage{amsmath}
\usepackage{url}
\usepackage{soul}
\usepackage{pgf,tikz}
\usepackage[linesnumbered,ruled]{algorithm2e} 

\newtheorem{theorem}{Theorem}
\newtheorem{lemma}{Lemma}
\newtheorem{corollary}{Corollary}
\newtheorem{proposition}{Proposition}
\newtheorem{remark}{Remark}

\newtheorem{example}{Example}
\newtheorem{assumption}{Assumption}
\newtheorem{definition}{Definition}
 \newcommand{\R}{{\mathbb{R}}}

\newcommand{\Z}{{\mathbb Z}}

\newcommand{\rhoo}{{\rho}}

\newcommand{\ie}{{\it i.e.}}

\usepackage{pgf,tikz}

\DeclareMathOperator{\Ker}{ker}
\DeclareMathOperator{\Image}{Im}

\DeclareMathOperator{\card}{card}
\DeclareMathOperator{\Rank}{rank}
\DeclareMathOperator{\GreedyStatic}{GreedyStatic}
\DeclareMathOperator{\GreedyTimeVarying}{GreedyTimeVarying}
\DeclareMathOperator{\tr}{\mathop{Trace}}
\DeclareMathOperator{\E}{\mathbf E}

\DeclareMathOperator{\diag}{diag}
\DeclareMathOperator{\DualSet}{DualSet}

\begin{document}
\title{Deterministic and Randomized Actuator Scheduling\\ With Guaranteed Performance Bounds}

\author{ M. Siami, A. Olshevsky,  and A. Jadbabaie\\ 
\thanks{This research was supported in part by a Vannevar Bush Fellowship from the Office of Secretary of Defense.}
\thanks{M. Siami is with the Department of Electrical and Computer Engineering, Northeastern University, Boston, MA 02115 USA (e-mail:  {\tt\small m.siami@northeastern.edu}). A. Olshevsky is with the Department of Electrical and Computer Engineering, Boston University, Boston, MA 02215 USA (e-mail:  {\tt\small alexols@bu.edu}). A. Jadbabaie is with the Institute for Data, Systems, and Society, Massachusetts Institute of Technology, Cambridge, MA 02139 USA (e-mail:  {\tt\small jadbabai@mit.edu}).}
}

\maketitle

\begin{abstract}
In this paper, we investigate the problem of actuator selection for linear dynamical systems. We develop a framework to design a sparse actuator schedule for a given large-scale linear system with guaranteed performance bounds using deterministic polynomial-time and randomized approximately linear-time algorithms.  First, we introduce systemic controllability metrics for linear dynamical systems that are monotone and homogeneous with respect to the controllability Gramian. We show that several popular and widely used optimization criteria in the literature belong to this class of controllability metrics. Our main result is to provide a polynomial-time actuator schedule that on average selects only a constant number of actuators at each time step, independent of the dimension, to furnish a guaranteed approximation of the controllability metrics in comparison to when all actuators are in use. Our results naturally apply to the dual problem of sensor selection, in which we provide a guaranteed approximation to the observability Gramian. We illustrate the effectiveness of our theoretical findings via several numerical simulations using benchmark examples.
\end{abstract}

\allowdisplaybreaks

\section{Introduction}

Over the past few years, controllability and observability properties of complex dynamical networks have been subjects of intense study in the controls community~\cite{Alex2014,pasqualetti2014controllability, liu2016control, 7962975,muller1972analysis, tzoumas2016minimal,7171062, segio2017,NozariACC,yaziciouglu2016graph,summers2016submodularity, pequito2015complexity}. This interest stems from the need to steer or observe the state of  large-scale, networked systems such as the power grids \cite{chakrabortty2011control}, social networks, biological and genetic regulatory networks \cite{ChandraBD11,marucci2009turn,rajapakse2012can}, and traffic networks \cite{SiamiTAC18}.
While the classical notion of  controllability and observability, introduced by  Kalman in \cite{kalman1963mathematical} is quite well understood, the dependence of various measures of controllability or observability on number and location of sensors and actuators in linear systems have been subject of study for nearly 5 decades \cite{ATHANS1972397}. Often times, there is a need to steer or estimate the state of a large-scale, networked control system with as few actuators/sensors as possible, due to issues related to cost and energy depletion. The desire to perform control/estimation using a sparse set of actuators/sensors spans various application domains, ranging from infrastructure networks (e.g., water and power networks) to multi-robot systems and the study of the human connectome. For example, energy conservation through efficient utilization of sensors and actuators can help extend the duration of battery life  in networks of mobile sensors and multi-agent robotic networks; estimating the whole state of the power grid using fewer measurement units will help reduce the cost of monitoring the network for systemic failures, etc.

It is therefore desirable to have a limited number of sensors and actuators without compromising the control or estimation performance too much. Unfortunately, as the recent works in~\cite{Alex2014,vassili1} have shown,  the problem of finding a sparse set of input variables such that the resulting system is controllable is NP-hard. Even the presumably easier problem of approximating the minimum number better than a constant multiplicative factor of $ \log n$ is also NP-hard. Other results in the  literature have studied network controllability by exploring approximation algorithms for the closely related subset selection problem \cite{Alex2014,summers2016submodularity, pequito2015complexity}. More recently, some of the authors showed that even the problem of finding a sparse set of actuators to guarantee reachability of a particular state is hard and even hard to approximate \cite{jadbabaie2017minimal}.

Previous studies have been mainly focused on solving the optimal sensor/actuator placement problem using the greedy heuristic, as approximations of the corresponding sparse-subset selection problem.
While these results attempt to find approximation algorithms for finding the best sparse subset, our focus in this paper is to gain new fundamental insights into approximating  various controllability metrics compared to the case when all possible actuators are chosen. Specifically, we are interested in actuator/sensor schedules that select a small number of actuators/sensors so as to save the energy while ensuring a suitable level of controllability (observability) performance for the entire network. Due to energy efficiency, we may want to reduce the number of active actuator/sensors at each time. At the same time, we would like to have a  performance that closely resembles that of the  original system, when all available sensor/actuators are active. 
 
We investigate sparse sensor and actuator selection as particular instances where discrete geometric structures can be utilized to study network controllability and observability problems (cf. \cite{vassili1,vassili2,vassili3}). A key observation is the  close connection between this problem and some classical problems in statistics such as outlier detection, active learning, and optimal experimental design. 
In recent years, there has been  a renewed interest in optimal experiment  design which has a long history going back at least 65 years \cite{kempthorne1952design, allen-zhu17e}.

{One of our main contributions is to show that the time-varying actuator selection problem, which goes back to a paper by Athans in 1972 \cite{ATHANS1972397}, can be solved via random sampling.} We propose an alternative to submodularity-based  methods and instead use recent advances in theoretical computer science to develop scalable algorithms for sparsifying control inputs. Current approaches based on polynomial time relaxations of the subset selection problem require an extra multiplicative factor of $\log n$ sensors/actuators times the minimal number in order to just maintain controllability/observability. Using these recent advances \cite{allen-zhu17e, lee2017sdp, Spielman, ramanujan, Rudelson:2007,marcus,srivastava2017alon}, we  show that by carefully designing a scheduling strategy, one can choose on average a constant number of sensors and actuators at each time, to approximate the controllability/observability metrics of the system when all sensors and actuators are in use.

{\bf Potential application domains.} One potential application can be considered as wide-area oscillation damping control using High Voltage DC (HVDC) lines (e.g., \cite{elizondo2018interarea, atawi2013advance,liu2017minimal}). HVDC systems are increasingly being installed in power grids all around the globe. This trend is expected to continue with recent advancements in power electronics technology, energy harvesting, and usage of renewable energy \cite{elizondo2018interarea}. In this setup, it seems quite compelling to examine approaches that can support sparse HVDC lines (\ie, actuators) scheduling to improve the controllability of the power grid in order to account for issues related to cost, energy depletion, and the limitations in directly accessing actuators, especially in large networks (cf. Example \ref{example_2} in Section \ref{sec:numerical}). Moreover, note that the dual problem of actuator scheduling for control is sensor scheduling for estimation. In the present case, our sparse sensor schedule setup is equal to reducing the number of measurements for data reduction, and the observability Gramian-based measures show how well one can estimate the state of the system \cite{boyd2008lecture}.

Another potential application is disease spread estimation in networks where testing resources are scarce.  There are several models for the spread of infections (see \cite{nowzari2016analysis} and references therein). 
Formally, let us consider a network with $n$ nodes. Each node represents a city and has a non-negative scalar state $x_i(k)$ associated with it, which indicates the prevalence of an infection in that node. Since $x_i(t)$ is the {\em proportion} of the population at node $i$ infected at time $t$, we assume that the $x_i(t)$ are close to zero (for example, this is valid for the recent COVID-19 pandemic; even though there are a substantial number of infected people still the proportion of the population that is infected is small as of early-April, 2020\footnote{\url{https://www.cdc.gov/coronavirus/2019-ncov/cases-updates/cases-in-us.html}}). It therefore makes sense to linearize the epidemic models around the zero state.

After linearization, the states evolve according to an autonomous linear differential equation $x(k+1) ~=~ A x(k)$; in all epidemic models, the off-diagonal entries $a_{ij}$ of the state matrix $A$ indicates the unitized transmission rate of the infection from city $j$ to city $i$, while the diagonal entries can be positive or negative, reflecting the possibility of either local spread or recovery.

We assume $y(t) = C(t) x+w(t)$, where $C(t)$ is a ``subset'' of the identity matrix (because one can measure the prevalence of the infection in a node by randomly testing from the population at that node) and $w(t)$ is  noise. The sensor scheduling problem in this context amounts estimating the state with as small a variance as possible, while the measurements have to be done over a certain time-horizon and are bounded in number due to scarce resources.\footnote{From now on we will focus the paper on the actuator selection problem. The dual notion of sensor selection follows similar ideas.}

Some of our results appeared earlier in the conference version of this paper \cite{ECC2018-Milad, ACC2018-Milad}; however, their proofs are presented here for the first time. The manuscript also contains several new results, remarks, numerical examples, and proofs.

\section{Preliminaries and Definitions}
\subsection{Mathematical Notations}
\label{sec:0}
Throughout the paper, discrete time index is denoted by $k$. The sets of real (integer), non-negative real (integer), and positive real (integer) numbers are represented by $\R$ ($\mathbb Z$), $\R_+$ ($\mathbb Z_+$) and $\R_{++}$ ($\mathbb Z_{++}$), respectively. The set of natural numbers $\{i \in \Z_{++} ~:~i \leq n\}$ is denoted by $[n]$. The cardinality of a set $\sigma$ is denoted by $\card(\sigma)$. Capital letters, such as $A$ or $B$, stand for real-valued matrices. 
For a square matrix $X$, $\det(X)$ and $\tr(X)$ refer to the determinant and the summation of on-diagonal elements of $X$, respectively. $\mathbb S^n_+$ is the positive definite cone of $n$-by-$n$ matrices. 
The $n$-by-$n$ identity matrix is denoted by $I$. Notation $A \preceq B$ is equivalent to matrix $B-A$ being positive semi-definite. The transpose of matrix $A$ is denoted by $A^\top$. The rank, kernel and image of matrix $A$ are referred to by $\Rank(A)$, $\ker(A)$ and $\Image(A)$, respectively.  
The Moore-Penrose pseudo-inverse of matrix $A$ is denoted by $A^{\dag}$.  The ceiling function of $x \in \R$ is denoted by $\lceil x\rceil$ where it returns the least integer greater than or equal to $x$. 

\subsection{Linear Systems and Controllability}

We start with the canonical linear discrete-time, time-invariant dynamics 
\[ x(k+1) ~=~A \, x(k) ~+~B \, u(k),\]
where $A \in \R^{n \times n}$, $B \in \R^{n \times m}$ and $k \in \Z_+$.
The state matrix $A$ describes the underlying structure of the system and the interaction strength between the agents, and matrix $B$ represents how the control input enters the system.
Equivalently, the dynamics can be written as 
\begin{equation}
 x(k+1) ~=~A\,x(k) ~+~\sum_{i \in [m]} b_i\, u_i(k),
 \label{model:a}
\end{equation}
where $b_i$'s are columns of matrix $B\in \R^{n \times m} $.
Then, the controllability matrix at time $t$ is given by
\begin{equation}
    \mathcal C(t)~=~\left [ B ~AB~A^2B~\cdots~ A^{t-1}B \right].
    \label{control-matrix}
\end{equation} 

In this paper, we assume that $t>0$ is the time horizon to control (also known as the time-to-control). It is well-known that  from a numerical standpoint it is better to characterize controllability in terms of the Gramian matrix at time $t$ defined as follows:
 \begin{equation}
\mathcal W(t) ~=~ \sum_{i=0}^{t-1} A^{i} BB^{\top} (A^{i})^\top ~=~ \mathcal C(t) \, \mathcal C^\top(t).
\label{gramian}
 \end{equation}
When looking at time-varying input/actuator schedules, we will consider the following linear system with time-varying input matrix $\mathcal B(.)$
\begin{equation}
 x(k+1) ~=~A \, x(k) ~+~ \mathcal B(k) \, u(k).
\label{g:model}
\end{equation}
For the above system, the controllability and Gramian matrices at time step $t$ are defined as
  \[ \mathcal C_*(t)~=~\big [ \mathcal B(t-1) ~A\mathcal B(t-2)~A^2\mathcal B(t-3)~\cdots ~A^{t-1}\mathcal B(0) \big], \]
  and
  \begin{eqnarray}
\mathcal W_*(t) &=& \sum_{i=0}^{t-1} A^{i} \mathcal B(t-i-1)\mathcal B^{\top}(t-i-1) (A^{i})^\top \nonumber \\
&=& \mathcal C_*(t) \, \mathcal C_*^\top(t),
\label{gramian-varying}
 \end{eqnarray}
 respectively.
 
\begin{assumption}
Throughout the paper, we assume that the system \eqref{model:a} is controllable (i.e., the controllability matrix has full row rank and the Gramian is positive definite). However, all results presented in this paper can be modified/extended to uncontrollable systems.
\end{assumption}

\subsection{Matrix Reconstruction and Sparsification}

The key idea through out the paper is to approximate the time-$t$ controllability Gramian as a sparse sum of rank-$1$ matrices, while controlling the approximation error. 
To this end, we present a key lemma from the sparsification literature and state the necessary modification. We then use this result later in our deterministic algorithm to find a sparse actuator schedule.

\begin{lemma}[Dual Set Spectral Sparsification \cite{Christos}]
\label{lemma-dualset}
Let $V=\{v_1, \ldots, v_t\}$ and $U=\{ u_1, \ldots, u_t\}$ be two equal cardinality decompositions of identity matrices (i.e., $\sum_{i=1}^t v_i v_i^\top = I_n$ and $\sum_{i=1}^t u_i u_i^\top = I_{\ell}$ ) where $v_i \in \R^n$ ($n < t$) and $u_i \in \R^\ell$ ($\ell \leq t$). Given an integer $\kappa$ with $n < \kappa \leq t$, Algorithm \ref{al::dualset} computes a set of weights $c_i \geq 0$ where $i \in [t]$, such that 
\[\lambda_{\min} \left(\sum_{i=1}^t c_i v_i v_i^\top  \right) \geq \left(1 - \sqrt{\frac{n}{\kappa}}\right)^2,\]
\[\lambda_{\max} \left(\sum_{i=1}^t c_i u_i u_i^\top  \right) \leq \left(1 + \sqrt{\frac{\ell}{\kappa}}\right)^2,\]
and
\[ \card \left \{ i ~:~ c_i >0, i \in [t] \right \} ~\leq~ \kappa.\]
\end{lemma}

{Due to space limitations, we refer the interested readers to \cite{Christos} for more details on Algorithm \ref{al::dualset}.}
However, roughly speaking, Algorithm \ref{al::dualset} is based on choosing vectors in a greedy fashion that satisfy a set of desired properties at each step, leading to bounds on eigenvalues.
{In Algorithm \ref{al::dualset}, {lower and upper} barriers or potential functions are defined as follows:
 \begin{equation}
 \underline{\phi} (\underline \mu , \underline{\mathcal A}) = \sum_{i=1}^n \frac{1}{\lambda_i(\underline{\mathcal A})-\underline \mu },
 \label{eq::269}
 \end{equation}
 and 
\begin{equation}
 \bar \phi (\bar \mu,\bar{\mathcal A}) ~=~ \sum_{i=1}^{ \ell} \frac{1}{\bar \mu -\lambda_i(\bar{\mathcal A})},
 \label{eq::274}
\end{equation}
respectively.}
These potential functions quantify how far the eigenvalues of $\underline{\mathcal A}$ and $\bar{\mathcal A}$ are from the barriers $\underline \mu$ and $\bar \mu$. These potential functions become unbounded as any eigenvalue nears the barriers.\footnote{These potentials are equal to constant multiples of the Stieltjes transform of $\underline{\mathcal A}$ and $\bar{\mathcal A}$ evaluated at $\underline \mu$ and $\bar \mu$, respectively \cite{silverstein2009stieltjes}.}  {We control the maximum eigenvalue of $\bar{\mathcal A}$ using an upper barrier $\bar \mu$ and the minimum eigenvalue of $\underline{\mathcal A}$ using a lower barrier $\underline{\mu}$.} Two parameters $\mathfrak L$ and $\mathfrak U$ are defined as follows:
\begin{small}
\begin{eqnarray*}
&&\hspace{-.6cm}\mathfrak L(v, \underline \delta,\underline{\mathcal A},\underline \mu)~= \nonumber \\
&&~~ \frac{ v^\top \left (\underline{\mathcal A}-(\underline \mu +\underline \delta)I_n \right )^{-2}v}{\underline \phi (\underline \mu+\underline \delta, \underline{\mathcal A}) - \underline \phi (\underline \mu, \underline{\mathcal A})}- v^\top \left (\underline{\mathcal A}-(\underline \mu +\underline \delta)I_n \right)^{-1}v,
\end{eqnarray*}
\end{small}
and
\begin{small}
\begin{eqnarray*}
&&\hspace{-.6cm}\mathfrak U(u, \bar \delta,\bar{\mathcal A},\bar \mu ) = \nonumber \\
&&~~ \frac{ u^\top ((\bar \mu +\bar \delta)I_{\ell} - \bar{\mathcal A})^{-2}u}{\bar \phi (\bar \mu, \bar{\mathcal A}) - \bar \phi (\bar \mu+\bar \delta, \bar{\mathcal A})}+u^\top \left ((\bar \mu +\bar \delta)I_{\ell} - \bar{\mathcal A} \right )^{-1}u.
\end{eqnarray*}
\end{small}
{The Sherman-Morrison-Woodbury formula inspires the structure of the above quantities for more details on the barrier method (cf. \cite[Section 1.2]{camacho2014spectral}).}
{These potential functions \eqref{eq::269} and \eqref{eq::274} are chosen to guide the selection of vectors and scalings at each timestep $\tau$ and to ensure steady progress of the algorithm.
Small values of these potentials indicate that the eigenvalues of $\bar{\mathcal A}$ and $\underline{\mathcal A}$ do not concentrate near $\bar{\mu}$ and $\underline{\mu}$, respectively. 
In Algorithm \ref{al::dualset}, at each iteration, we increase the upper barrier $\bar{\mu}$ by a fixed constant $\bar{\delta}$ and the lower barrier $\bar{\mu}$ by another fixed constant $\underline{\delta}$. It can be shown that as long as the potentials remain bounded, there must exist (at every step $\tau$) a choice of an index $j$ and weight $c_j$  so that the addition of associated rank-1 matrices to $\bar{\mathcal A}$ and $\underline{\mathcal A}$, and the increments of barriers do not increase either potential and keep all the eigenvalues of the updated matrix between the barriers (see Algorithm \ref{al::dualset}). Repeating these steps ensures steady growth of all the eigenvalues and yields the desired result.}

This algorithm is a generalization of an algorithm from \cite{ramanujan} which is deterministic and at most needs $\mathcal O\left( \kappa t (n^2+\ell^2)\right)$. Furthermore, the algorithm needs $\mathcal O(\kappa tn^2)$ operations if $U$ contains the standard basis of $\R^t$; we refer the reader to \cite{Christos} for more details.

 \begin{algorithm}[t]
  {\small
    \SetKwInOut{Input}{Input}
    \SetKwInOut{Output}{Output}

    \Input{$V=\left [v_1,\ldots, v_t\right]$, with $VV^\top=I_n$\\
    $U=\left [u_1,\ldots, u_t\right]$, with $UU^\top=I_\ell$\\
    $\kappa \in \Z_+$, with $n<\kappa\leq t$\\ }
            \vspace{.1cm}
    \Output{$c=[c_1,c_2, \ldots, c_t] \in \R^{1 \times t}_+$ with $\|c\|_{0} \leq \kappa$ }
    		\vspace{.4cm}
Set $c(0)=0_{t \times 1}$, $\underline{\mathcal A}(0)=0_{n \times n}$, $\bar{\mathcal A}(0)=0_{\ell \times \ell}$, $\underline \delta=1$, $\bar \delta=\frac{1+\sqrt{\frac{\ell}{\kappa}}}{1-\sqrt{\frac{n}{\kappa}}}$\\
        \vspace{.1cm}
\For{$\tau=0:\kappa-1$}        
{
$\underline \mu({\tau})= \tau - \sqrt{\kappa n}$\\
$\bar \mu({\tau})=\bar \delta \left( \tau+\sqrt{\kappa\ell}\right)$\\
Find an index $j$ such that
\[ \mathfrak U(u_j, \bar \delta, \bar{\mathcal A}(\tau), \bar \mu({\tau})) \leq \mathfrak L(v_j, \underline \delta, \underline{\mathcal A}({\tau}), \underline \mu({\tau}))\]\\
Set $\Delta ={2}{\left(\mathfrak U(u_j, \bar \delta, \bar{\mathcal A}(\tau), \bar \mu({\tau})) + \mathfrak L(v_j, \underline \delta, \underline{\mathcal A}({\tau}), \underline \mu({\tau}))\right)^{-1}}$\\
Update the $j$-th component of $c(\tau)$:
\[c({\tau+1})=c(\tau)+ \Delta  \text{e}_j,\]\\
$\underline{\mathcal A}({\tau+1}
)=\underline{\mathcal A}({\tau})+\Delta  v_jv_j^\top$\\
$\bar{\mathcal A}({\tau+1})=\bar{\mathcal A}({\tau})+ \Delta u_ju_j^\top$
}
        \vspace{.1cm}
\Return $c=\kappa ^{-1}\left(1- \sqrt{\frac{n}{\kappa}}\right) c(\kappa)$\\

        \vspace{.1cm}
    \caption{\small A Deterministic Dual Set Spectral Sparsification $\DualSet(V, U, \kappa)$.}
    \label{al::dualset}}
\end{algorithm}

{\begin{remark}
We modify the $5^{\text{th}}$ line of Algorithm \ref{al::dualset}; at each step, we choose an index $j$ that maximizes 
\begin{equation} 
	\mathfrak L(v_j, \underline \delta, \underline{\mathcal A}({\tau}), \underline \mu({\tau}))- \mathfrak U(u_j, \bar \delta, \bar{\mathcal A}(\tau), \bar \mu({\tau})), 
	\label{eq::dualset1}
\end{equation}
instead of only finding an index $j$ such that 
\begin{equation}
	\mathfrak U(u_j, \bar \delta, \bar{\mathcal A}(\tau), \bar \mu({\tau})) \leq \mathfrak L(v_j, \underline \delta, \underline{\mathcal A}({\tau}), \underline \mu({\tau})).
	\label{eq::dualset2}
\end{equation}
We should note that if an index $j$ maximizes \eqref{eq::dualset1}, then it will satisfy \eqref{eq::dualset2}. Therefore, Lemma \ref{lemma-dualset} still holds for the modified algorithm, and hence the theoretical bounds are valid. {Based on our simulations, we observe that this modification can help to improve Algorithm \ref{al::dualset} by producing smaller ratio $\lambda_{\max}\left(\sum_{i=1}^t c_i v_i v_i^\top  \right)/\lambda_{\min}\left(\sum_{i=1}^t c_i v_i v_i^\top  \right)$ (in Section \ref{sec:weighted}, we will see that this quantity is closely related to approximation factor $\epsilon$).}
We denote the application of the algorithm to $V$ and $U$ by
\[ [c_1, c_2, \cdots, c_t] = \DualSet^* (V, U, \kappa).\]
\end{remark}
}

{ We now recall the concentration lemma of Rudelson-Vershynin \cite{Rudelson:2007} as follows. We are going to use this result in the proof of Theorem \ref{th:rand}.
\begin{lemma} \cite[Thm. 3.1]{Rudelson:2007} 
\label{rudelson}
Let $y \in \R^p$ be a random vector such that $\|y\| \leq b$ almost surely and $\|\E yy^\top \|_2 \leq 1$. Let $y_1,\cdots, y_n$ be i.i.d. copies of $y$. Then 
\begin{equation}
\E \left \|\frac{1}{n} \sum_{i=1}^n y_i y_i^\top-\E y y^\top \right \|_2 \leq \min \left (1, cb\sqrt{\frac{\log n}{n}}\right ),
\end{equation}
where $c > 0$ is some universal constant.
\end{lemma}
}
In the next section, we show how various controllability measures can be approximated by selecting a sparse set of actuators via the above algorithm.

\section{Systemic Controllability Metrics}

\begin{table*}[t]
{\small
  \begin{center}
  \begin{tabular}{|c|l|c|}
  \hline
    Optimality-criteria & Systemic Controllability Measure & Matrix Operator Form
    \\
    \hline
    \hline
    A-optimality & Average control energy & $\tr \left (\mathcal W^{-1}(t)\right)$   
    \\  
    \hline
    D-optimality & The volume of the ellipsoid & $\left (\det \mathcal W(t) \right)^{-1/n}$ 
    \\
    \hline
    T-optimality & Inverse of the trace & $1 /\tr (\mathcal W(t)) $
    \\
    \hline
    E-optimality & Inverse of the minimum eigenvalue & $1/\lambda_{\text{min}} (\mathcal W(t))$ 
    \\
    \hline
\end{tabular}
\end{center}}
    \caption{ \small Some important examples of systemic controllability metrics.} 
\label{Table:PerformanceMeasures}
\end{table*}

Similar to the {\it systemic} notions introduced in \cite{siami2017abstraction, Siami, siami2017growing}, we define various controllability metrics. These measures are  real-valued operators defined on  the set of all linear dynamical systems governed by  (\ref{g:model}) and quantify various measures of the required control energy.   All  of the metrics depend on the controllability Gramian matrix of the system which is a positive definite matrix. Therefore, one can define a systemic controllability performance measure as an operator on the set of Gramian matrices of all controllable systems with $n$ states which we represent by $\mathbb S_+^n$.\footnote{For any $X \in \mathbb S_+^n$ and given $t  \in \mathbb Z_{++}$, there exists at least one controllable system with $\mathcal W(t) = X$ (e.g., $x(k+1) = X^{\frac{1}{2}} u(k)$), and for any controllable system, it is well known that the Gramian matrix is positive definite (see \eqref{gramian-varying}). Therefore, the set of Gramian matrices of all controllable systems with $n$ states is equal to $\mathbb S_+^n$.    }

\begin{definition}[Systemic Criteria]
\label{metric}
A controllability metric $\rho: \mathbb S_+^{n} \rightarrow \R$ is systemic if and only if \\

	\noindent {1.} {\it Homogeneity:} For  all $\kappa >1$, 
	\[ \rhoo (\kappa A)~=~\kappa^{-1} \rhoo (A);\]
	\noindent {2.} {\it Monotonicity:} If $B \preceq A$, then
				\[\rhoo (A) ~\leq~ \rhoo (B).\]	
\end{definition}

For many popular choices of $\rho$, one can see that they satisfy the properties presented in Definition \ref{metric}. Some of them are listed in Table \ref{Table:PerformanceMeasures}. We note that similar criteria have been developed \cite{ravi2016experimental,kempthorne1952design,allen-zhu17e} in the experiment design literature (cf. Table I). In what follows, we will make this statement formal.

\begin{proposition}
For given dynamics {\eqref{g:model}} with Gramian matrix $\mathcal W(t)$, the metrics presented in Table 1 are systemic controllability measures.
\end{proposition}

\begin{proof}
One can easily see that all these measures satisfy the homogeneity, and monotonicity properties in Definition \ref{metric} (cf. \cite{pukelsheim1993optimal, siami2017growing}).
\end{proof}

In the next section, we show how various measures can be approximated by selecting a sparse set of actuators.

\section{Sparse Actuator Selection Problems}
\label{sec:sparse}
For a given linear system \eqref{model:a} with a general underlying structure, the actuator scheduling problem seeks to construct a schedule of the control inputs that keeps the number of active actuators much less than the original system such that the controllability matrices of the original and the new systems are similar in an appropriately defined sense.
Specifically, given a canonical linear, time-invariant system \eqref{model:a} with $m$ actuators and controllability Gramian matrix $\mathcal W(t)$ at time $t$, our goal is to find a sparse actuator schedule such that the resulting system with controllability Gramian $\mathcal W_s(t)$ is well-approximated, i.e., 
\begin{equation}
\left| \frac{\rho\left (\mathcal W(t)\right)-\rho\left (\mathcal W_s(t)\right )}{\rho \left (\mathcal W(t) \right )} \right| ~\le~ \epsilon,
\end{equation}
where $\rho$ is any systemic controllability metric that quantifies the difficulty of the control problem for example as a function of the required control energy, and $\epsilon \geq 0 $ is the approximation factor. The systemic controllability metrics are defined based on the controllability Gramian, therefore ``close" Gramian matrices result in approximately the same values.
Our goal here is to answer the following questions:

\begin{itemize}
    \item[-] What is the minimum number of actuators to be chosen to achieve a good approximation of the system with the full set of actuators utilized?
    \item[-] What is the relation between the number of selected actuators and performance/controllability loss?
    \item[-] Does a sparse approximation schedule exist with at most a constant number of  active actuators at each time?
    \item[-] What is the time complexity of choosing the subset of actuators with guaranteed performance bounds?
\end{itemize}

In the rest of this paper, we show how some fairly recent advances in theoretical computer science and the probabilistic method can be utilized to answer these questions. The probabilistic method is one of the most important tools of modern combinatorics which was introduced by Erd\"os. The idea is that a deterministic solution is shown to exist by constructing a random candidate satisfying all the requirements of the problem with positive probability. Recently, Marcus, Spielman, and Srivastava introduced a new variant of the probabilistic method which ends up solving the so-called Kadison-Singer (KS) conjecture \cite{marcus}.  We use the solution approach to the KS conjecture together with  a combination of  tools from  Sections \ref{sec:weighted} to find a sparse approximation of the actuator selection problem with algorithms that have favorable time-complexity.

\section{A Weighted Sparse Actuator Schedule}
\label{sec:weighted}

As a starting point, we allow for scaling of the input signals at chosen inputs while keeping the input scaling bounded. The input scaling allows for an extra degree of freedom that could allow for choosing a sparser set of inputs. 
Given \eqref{model:a}, we define a weighted actuator schedule by $\sigma = \{\sigma_k\}_{k=0}^{t-1}$ and scalings $s_i(k) \geq 0$ where  $i \in [m]$, $k+1 \in [t]$, and $\sigma_k=\{i | s_i(k) >0\} \subseteq [m]$. The resulting system with this schedule is 
\begin{equation}
 x(k+1) ~=~A \, x(k) ~+~\sum_{i \in \sigma_k} s_{i}(k)\, b_i \, u_i(k),~k \in \Z_+
 \label{model:b}
 \end{equation}
where $s_{i}(k) \geq 0$ shows the strength of the $i$-th control input at time $k$. 
The controllability Gramian \eqref{gramian-varying} at time $t$ for this system can be rewritten as
\begin{equation}
\mathcal W_{s}(t) ~=~ \sum_{k=0}^{t-1} \sum_{j \in \sigma_k} s_j^2(k)\left (A^{t-k-1} b_j\right)\left (A^{t-k-1}b_j\right)^\top.
\label{W_s}
\end{equation}

Our goal is to reduce the number of active actuators {\em on average} $d$, where  
\begin{small}
\begin{equation}
 d~:=~\frac{ \sum_{k=0}^{t-1}\card\left \{\sigma_k\right \}}{t},
 \label{def-d}
\end{equation}\end{small}such that the controllability Gramian of the fully actuated and the new sparsely actuated system  are ``close." Of course, this approximation will require horizon lengths that are potentially longer than the dimension of the state. The definition below formalizes this approximation.

\begin{definition}
\label{def:aprox}
Given a time horizon $t \geq n$, system \eqref{model:b} with a weighted actuator schedule is $(\epsilon, d)$-approximation of system \eqref{model:a} if and only if
\begin{equation}
(1-\epsilon)\,\mathcal W(t) ~\preceq ~ \mathcal W_s(t) ~\preceq ~(1+\epsilon)\,\mathcal W(t),
\label{eq:321}
\end{equation}
where $\mathcal W(t)$ and $\mathcal W_s(t)$ are the controllability Gramian matrices of \eqref{model:a} and \eqref{model:b}, respectively, and  parameter $d$ is defined by \eqref{def-d} as the average number of active actuators, and $\epsilon \in (0,1)$ is the approximation factor.
\end{definition}

\begin{remark}
 While it might appear that allowing for the choice of $s_i(k)$ might lead to amplification of input signals,  we  note that the scaling  cannot be too large because  the approximation is two-sided. Specifically, by taking  the trace from both sides of \eqref{eq:321}, we can see that the weighted summation of  $s_i^2(k)$'s is bounded. Moreover, based on Definition \ref{def:aprox}, the ranks of matrices $\mathcal W(t)$ and $\mathcal W_s(t)$ are the same. Thus, the resulting $(\epsilon,d)$-approximation remains controllable (recall that we assume that the original system is controllable).
\end{remark}

{
\begin{remark}
 The results presented in this paper also work for the case of linear time-varying systems, and it is straightforward to extend them for affine nonlinear discrete-time systems as well.
\end{remark}
}

\subsection{Deterministic Approach: Sparsifying Sums of Rank-one Matrices}
 The next theorem constructs a solution for the sparse weighted actuator schedule problem in polynomial time.

 \begin{algorithm}[t]
 {\small
    \SetKwInOut{Input}{Input}
    \SetKwInOut{Output}{Output}

    \Input{$A \in \R^{n \times n}$, $B \in \R^{n \times m}$, $t$ and $d$}
            \vspace{.1cm}
    \Output{$s_i(k) \geq 0$ for $(i,k+1) \in [m] \times [t]$}
    		\vspace{.4cm}
$\mathcal C(t) := \left [ B~AB~A^2B~\cdots~A^{t-1}B \right ]$\\
        \vspace{.1cm}
Set $V= \left(\mathcal C(t) \mathcal C^\top(t)\right)^{-\frac{1}{2}}\mathcal C(t)$\\
        \vspace{.1cm}
Set $U=V$\\
        \vspace{.1cm}
Run $[c_1, \cdots, c_{mt}] = \DualSet^* (V, U,dt)$\\
        \vspace{.1cm}
\Return $s_i(k):=\sqrt{c_{i+mk}{/(1+\frac{n}{dt})}}$ for $(i,k+1) \in [m] \times [t]$\\
        \vspace{.1cm}
\caption{\small A deterministic greedy-based algorithm to construct a sparse weighted actuator schedule (Theorem \ref{th-main-approx}).}
    \label{al-approx}}
\end{algorithm}

 \begin{algorithm}[t]
 {\small
    \SetKwInOut{Input}{Input}
    \SetKwInOut{Output}{Output}

    \Input{$A \in \R^{n \times n}$, $B \in \R^{n \times m}$, $t$ and $d$}
            \vspace{.1cm}
    \Output{$s_i(k) \geq 0$ for $(i,k+1) \in [m] \times [t]$}
    		\vspace{.4cm}
$\mathcal C(t) := \left [ B~AB~A^2B~\cdots~A^{t-1}B \right ]$\\
        \vspace{.1cm}
Set $V= \left(\mathcal C(t) \mathcal C^\top(t)\right)^{-\frac{1}{2}}\mathcal C(t)$\\
        \vspace{.1cm}
Set \[U~=~ \left[\underbrace{\text e_1, \ldots, \text e_{mt}}_{=I_{mt}}\right]\]
\tcp{where $\text e_i \in \R^{mt}$ for $i \in [mt]$ are the standard basis vectors for $\R^{mt}$}
        \vspace{.1cm}
Run $[c_1, \cdots, c_{mt}] = \DualSet^* (V, U,dt)$\\
        \vspace{.1cm}
\Return $s_i(k):= \sqrt{c_{i+mk}}$ for $(i,k+1) \in [m] \times [t]$\\
        \vspace{.1cm}
\caption{\small A deterministic greedy-based algorithm to construct a sparse weighted actuator schedule (Corollary \ref{th-main-2}).}
    \label{al-max-1}}
\end{algorithm}

\vspace{.01cm}
\begin{theorem}
\label{th-main-approx}
Given the  time horizon $t \geq n$, model \eqref{model:a}, and $d > 1$,  
Algorithm \ref{al-approx} deterministically constructs an actuator schedule such that the resulting system \eqref{model:b} is a  $\left (\epsilon,d \right )$-approximation of \eqref{model:a} with $\epsilon=\frac{2}{\sqrt {\frac{dt}{n}} +\sqrt\frac{n}{dt}}$ in at most $\mathcal O\left ( dm(tn)^2 \right)$ operations.
\end{theorem}

\begin{proof}
The controllability Gramian  of \eqref{model:a} at time $t$ is given by
\begin{eqnarray}
\mathcal W(t)&=&\sum_{i=0}^{t-1} \sum_{j=1}^m (\underbrace{A^{i} b_j}_{v_{ij}})(A^{i} b_j)^\top \nonumber  \\ 
&=& \sum_{i=0}^{t-1} \sum_{j=1}^m v_{ij} v_{ij}^\top.
\label{eq:465}
\end{eqnarray}
By multiplying $\mathcal W^{-\frac{1}{2}}(t)$ on both sides of \eqref{eq:465}, it follows that 
\begin{eqnarray}
I&=&\sum_{i=0}^{t-1} \sum_{j=1}^m \underbrace{(\mathcal W^{-\frac{1}{2}}(t)A^{i} b_j)}_{\bar v_{ij}} \, (\mathcal W^{-\frac{1}{2}}(t) A^{i} b_j )^\top \nonumber \\
&=& \sum_{i=0}^{t-1} \sum_{j=1}^m \bar  v_{ij} \bar v_{ij}^\top. 
\label{sum::identity}
\end{eqnarray}
Next, we define $U := \{ \bar v_{ij} | i+1 \in [t], j \in [m]\}$ and $V := U$. According to \eqref{gramian}, \eqref{eq:465}, and \eqref{sum::identity}, elements of $U$ are the columns of matrix $\left(\mathcal C(t) \mathcal C^\top(t)\right)^{-\frac{1}{2}}\mathcal C(t)$.
We now apply Lemma \ref{lemma-dualset}, which shows that there exist scalars $\bar c_{ij} \geq 0$ with 
\begin{equation}
    \card \left \{(i,j)\,:\,i+1 \in [t], j \in [m] , \bar c_{ij}  > 0\right \} ~\leq~ \frac{dt}{n}\times n,
    \label{eq:433}
\end{equation}
such that
\begin{equation*}
 \left ( 1- \sqrt \frac{n}{dt} \right)^2 I ~\preceq ~  \sum_{i=0}^{t-1}\sum_{j=1}^m \bar c_{ij} \, \bar  v_{ij} \, \bar v_{ij}^\top,
\end{equation*}
and
\begin{equation*}
\sum_{i=0}^{t-1}\sum_{j=1}^m \bar c_{ij} \, \bar  v_{ij} \, \bar v_{ij}^\top~\preceq ~ \left ( 1+ \sqrt \frac{n}{dt} \right)^2 I,
\end{equation*}
 or equivalently, 
\begin{equation}
\small{
\left ( 1- \sqrt \frac{n}{dt} \right)^2 \mathcal W(t) \preceq  \sum_{i=0}^{t-1}\sum_{j=1}^m \bar c_{ij} \, v_{ij} \, v_{ij}^\top \preceq  \left ( 1+ \sqrt \frac{n}{dt} \right)^2 \mathcal W(t).}
\label{eq:441}
\end{equation}
We can of course write  the controllability Gramian  of \eqref{model:b} at time $t$ as
\begin{eqnarray*}
\mathcal W_s(t)&=&\sum_{i=0}^{t-1} \sum_{j=1}^m s_j^2(t-i-1) (\underbrace{A^{i} b_j}_{v_{ij}})(A^{i} b_j)^\top \\
&=& \sum_{i=0}^{t-1} \sum_{j=1}^m s_j^2(t-i-1) \, v_{ij} v_{ij}^\top. 
\end{eqnarray*}
Define $\epsilon~:=~\frac{2}{\sqrt {\frac{dt}{n}} +\sqrt\frac{n}{dt}}$,
and 
\begin{equation}
       s_j(t-i-1):=\sqrt{\bar c_{ij} /(1+\frac{n}{dt})}.
\end{equation}
Then, by substituting $ \left ( 1+ \frac{n}{dt}\right)s_j^2(t-i-1)$ for $\bar c_{ij}$ in \eqref{eq:441}, we get
\begin{equation}
(1 - \epsilon) \mathcal W(t) ~\preceq ~ \mathcal W_s(t) ~\preceq ~  (1+\epsilon) \mathcal W(t).
\label{eq:452a}
\end{equation}
Finally, using \eqref{eq:452a}, \eqref{eq:433}, and Definition \ref{def:aprox}, we obtain the desired result. Moreover, this algorithm  runs in $dt$ iterations; In each iteration,  the functions $\mathfrak U$ and $\mathfrak L$ are evaluated at most $mt$ times. All $mt$ evaluations for both functions need at most $\mathcal O(n^3+mtn^2)$ time, because for all of them the
matrix inversions and {eigenvalue decompositions} can be calculated once. Finally, the updating step needs an additional $\mathcal O(n^2)$ time. Overall, the complexity of the algorithm is
of the order  $\mathcal O\left ( dm(tn)^2 \right)$. 
\end{proof}



\begin{remark}
For a given $d \geq 1$, while choosing $dt$ columns of the controllability matrix that form a full row rank matrix (i.e., the system is controllable) is an easy task but finding $dt$ columns of the controllability matrix that approximate the full Gramian matrix is what we are interested in here.
To do so, we should note that approximating the full Gramian matrix while keeping the number of active actuators less than a constant $d$ at each time is not possible in general. For example, in the case that $A=\mathbf 0_{n \times n}$ and $B=I_n$, at least all actuators at time $k=t-1$ are needed to form a full row rank matrix (or to approximate the full Gramian matrix). However, as we mentioned earlier, the number of active actuators on average can be kept constant in order to approximate the full Gramian matrix. Furthermore, condition $dt \geq n$ is needed for any algorithm that has a hope of success. Indeed, taking $B = I_n$ and $A = I_n$, it is straightforward to see that if $dt < n$, then we cannot hope to approximate the controllability Gramian because the controllability matrix of any schedule with $d$ active actuators on average is not full rank.
\end{remark}

{\subsubsection*{Feasibility Results}

The next remark uses results from the graph sparsification literature to show that the obtained bound on $\epsilon$ in Theorem \ref{th-main-approx} is within a constant factor of optimal.

	\begin{figure}[t]
	\centering      
	\includegraphics[trim = 0.5 0.5 0.5 0.5, clip,width=.4 \textwidth]{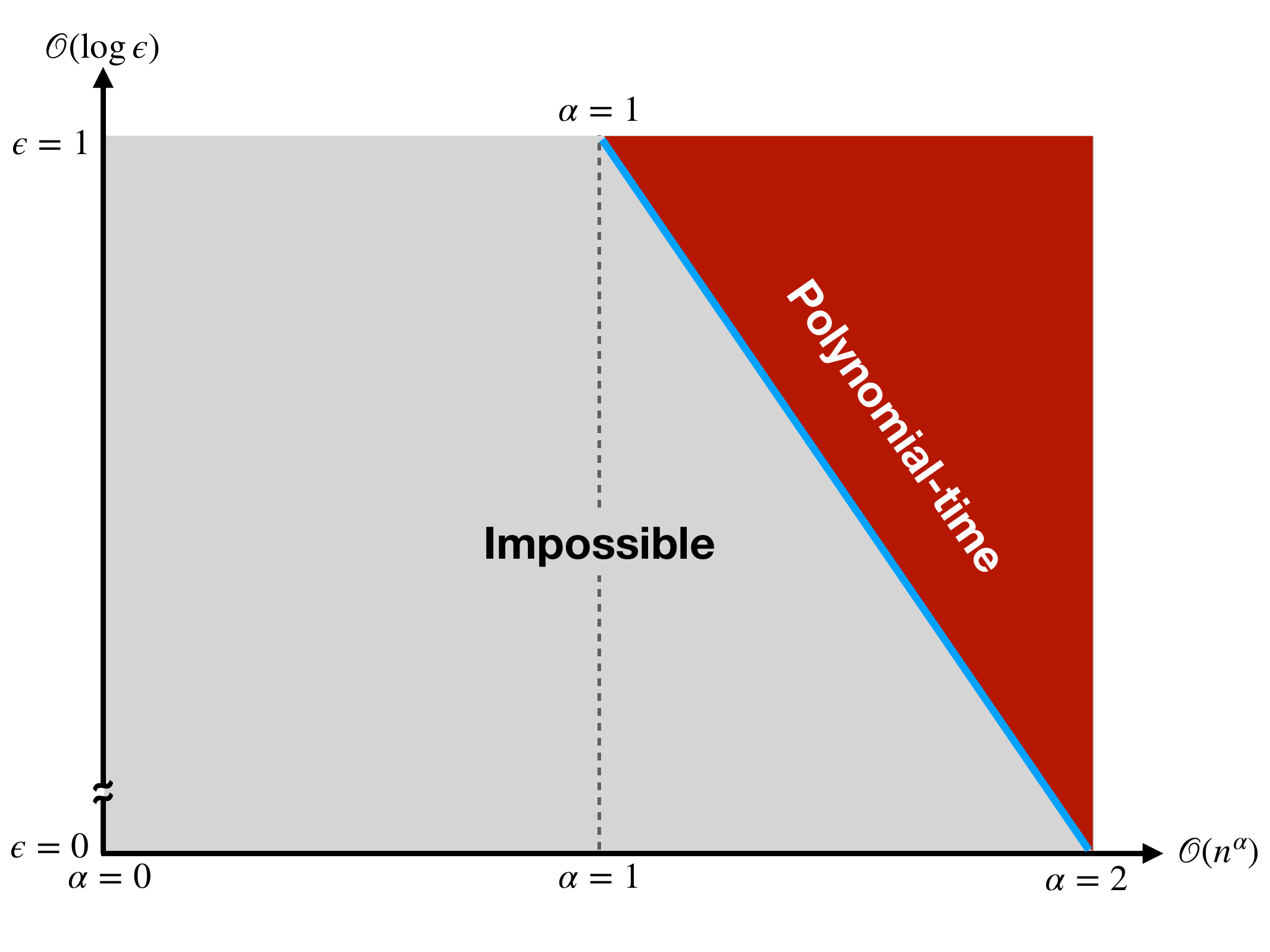} 
	\caption{\small  An illustration of how $\log \epsilon$ and the total number of active actuators $td=\mathcal O(n^\alpha)$ are related to each other (for simplicity we assume $t=m=n$). In the polynomial-time regime (i.e., red area), there is a polynomial-time algorithm for finding $(\epsilon, d)$-approximation. In the gray area, the sparse actuator scheduling is not feasible (see Remark \ref{remark:5}).  The blue line shows the presented bounds in Theorem \ref{th-main-approx}. In this plot, it is assumed that $n \rightarrow \infty$.}
	\label{fig:complex}
\end{figure}

\begin{remark}
\label{remark:5}
{In this remark, {we assume $t=n$, which means the time-to-control is $n$ steps.} We {present} a class of linear systems such that there is no sparse actuator schedule that achieves $(d,\epsilon)$-approximation where $\epsilon < \frac{1}{\sqrt{d/2}+1}$ for arbitrarily large $n$. {We will see that if we keep $d$ fixed while we let $n$ grow, $\epsilon$ cannot exceed this bound in the limit by translating our problem to an equivalent solved problem in graph theory.}
There is a close connection between sparse actuator scheduling and graph sparsification. The contribution of each link in the Laplacian matrix is a rank-one matrix same as the contribution of each actuator at each time in the Gramian matrix. Therefore, to sparsify a graph (\ie, reducing the number of links), we need to sparsify the summation of rank-one matrices. In this paper, we use a similar strategy to sparsify actuators in time and space. 

Consider model \eqref{model:a} with the following state matrices
\[A = \big [\textrm{e}_2, \textrm{e}_3, \cdots,  \textrm{e}_n, \textrm{e}_1 \big ]
= \begin{bmatrix}
	0&0&0&\cdots&0&1\\
 	1&0&0&\cdots&0&0\\
 	0&1&0&\cdots&0&0\\
  	\vdots&\vdots&\vdots&\ddots&\vdots&\vdots\\
	0&0&0&\cdots&0&0\\
   	0&0&0&\cdots&1&0
\end{bmatrix},
\]
and 
\[B =  \begin{bmatrix} 
	\sqrt{2n}&-1&-1&\ldots&-1\\
 	0&1&0&\cdots&0\\
 	0&0&1&\cdots&0\\
  	\vdots&\vdots&\vdots&\ddots&\vdots\\
   	0&0&0&\cdots&1
\end{bmatrix},\]
where $\textrm{e}_i$'s are the standard basis for $\R^n$.
For this system, the controllability has a unique structure: the columns of the controllability matrix \eqref{control-matrix} are the union of  $\left \{\sqrt{2n} \textrm{e}_i |~i \in [n]\right \}$ and all columns of the incidence matrix of a complete graph with $n$ nodes.\footnote{The oriented incidence matrix of a directed graph has one row for each node and one column for each link of the graph. If a link runs from node $i$ to node $j$, the column corresponding to that link has $-1$ in row $i$ and $1$ in row $j$; all other entries in that column are $0$.}  We now explain why: first, note that the columns of matrix $B$ are $\sqrt{2n} \textrm{e}_1$ and columns of the incidence matrix of a star graph with $n$ nodes where the center node is node $1$; moreover, $A$ is an elementary matrix that is obtained by shifting rows of the identity matrix. Therefore, for the other blocks  of the controllability matrix (i.e., $A^iB$ where $i=1, \ldots, n-1$), the columns are $\sqrt{2n} \textrm{e}_{i+1}$ and columns of the incidence matrix of a star graph with $n$ nodes where the center node is node $i+1$ (since $A^i$ is a permutation matrix). Note that links of a complete graph with $n$ nodes consists of directed links of $n$ star graphs with $n$ nodes. 
From the graph sparsification literature, the method of \cite{ramanujan} proves a more general result not only for graph sparsification but also about sparsifying sums of rank-one matrices. {By repeating the steps of \cite[Proposition 4.2]{ramanujan}, it now follows that} for this class of linear systems there is no weighted actuator schedule that can achieve $\epsilon$ less than 
\[\frac{{1}}{\sqrt{\frac{d}{2}}+1},\]
for arbitrarily large $n$ and given bounded $d$ greater than two (cf. \cite{srivastava2018alon}).  This means that $\epsilon$ cannot be reduced below $\mathcal O\left ( \frac{2}{\sqrt d + \sqrt{d^{-1}}} \right )$, where $\frac{2}{\sqrt d + \sqrt{d^{-1}}}$ is the proposed lower bound in Theorem \ref{th-main-approx} where $t=n$. Therefore, we conclude that our obtained bound on $\epsilon$ is within a constant  factor of optimal.}
\end{remark}

{ Fig. \ref{fig:complex} presents a computational phase diagram for the sparse actuator scheduling problem. This plot shows how $\log \epsilon$ and the total number of active actuators $td=\mathcal O(n^\alpha)$ are related to each other (for simplicity we assume $t=m=n$). In the polynomial-time regime (i.e., red area), there is a polynomial-time algorithm for finding $(\epsilon, d)$-approximation. In the gray area, the sparse actuator scheduling is not feasible (see Remark \ref{remark:5}).  The blue line shows the presented bounds in Theorem \ref{th-main-approx}. In this plot, it is assumed that $n \rightarrow \infty$.
}

\subsubsection*{Tradeoffs}
Theorem \ref{th-main-approx} illustrates a tradeoff between the average number of active actuators $d$ and the time horizon $t$ (also known as the time-to-control). 
This implies that the reduction in the  average number of active actuators comes at the expense of  increasing time horizon $t$ in order to get the same approximation factor $\epsilon$. 
Moreover, the approximation becomes more accurate as  $t$ and $d$ are increased. Of course, increasing $d$ will require more active actuators and  larger $t$ requires a larger control time window.

Fig. \ref{fig:1} depicts the approximation ratio $\epsilon$ given by Theorem \ref{th-main-approx} versus the average number of active actuators $d$ and the normalized time horizon ${t}/{n}$. We note that the approximation factor improves as $t$ become larger than $n$. Moreover, because of $\frac{2}{x+\frac{1}{x}} \leq 1$ for $x>0$, the approximation factor $\epsilon=\frac{2}{\sqrt {\frac{dt}{n}} +\sqrt\frac{n}{dt}}$ is always less than or equal to one. Hence, the upper bound ratio in \eqref{eq:321} is at most two.

	\begin{figure}[t]
	\centering
	 \psfrag{x}[h][h]{ \footnotesize{$d$}}  
	 \psfrag{y}[t][t]{ \footnotesize{$t/n$}}  
	 \psfrag{z}[t][t]{ \footnotesize{$\epsilon$}}        
	\includegraphics[trim = 0.5 0.5 0.5 0.5, clip,width=.4 \textwidth]{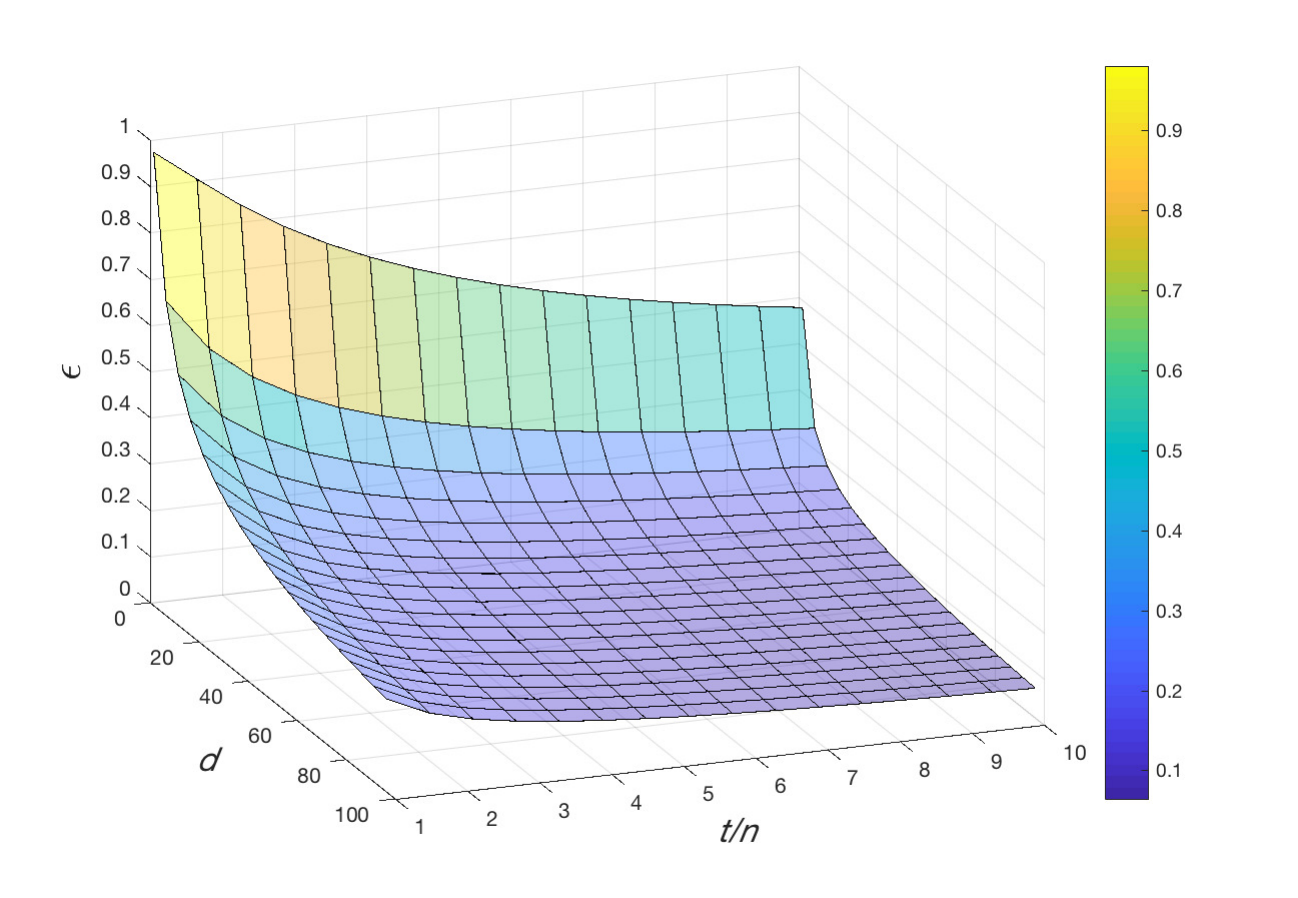} 
	\caption{\small This plot presents the approximation factor $\epsilon$ given by Theorem \ref{th-main-approx} versus the average number of active actuators $d \in (1, 100]$ and the normalized time horizon ${t}/{n}$. 	 }
	\label{fig:1}
\end{figure}

\subsubsection*{Sparse Actuator Schedules with Energy Constraints}

In this subsection, based on the energy/budget constraints on the scalings $s_i(k)$'s where $i \in [m]$ and $k+1 \in [t]$; three cases are considered as follows
\begin{itemize}
    \item[(i)] the scaling ratios are bounded, i.e., \[\max_{i \in [m],k+1 \in [t]} s^2_i(k)~\leq~\gamma,  \] 
    \item[(ii)] the sum of scaling ratios for each input is bounded, i.e., \[\max_{i \in [m]} \sum_{k+1 \in [t]} s_i^2(k)~\leq~\gamma,\]
    \item[(iii)] the sum of scaling ratios at each time is bounded, i.e.,  \[\max_{k+1 \in [t]} \sum_{i \in [m]} s_i^2(k)~ \leq~ \gamma. \]
\end{itemize}
%
In the next corollaries, we present deterministic sparse actuator schedules with the above energy/budget constraints. {These corollaries trade one of the inequalities in Theorem \ref{th-main-approx} with a single fixed bound on the size of scalings.}

\begin{corollary}
\label{th-main-2}
Given the  time horizon $t \geq n$, model \eqref{model:a}, and $d > 1$,  
Algorithm \ref{al-max-1} deterministically constructs an actuator schedule for \eqref{model:b} in at most $\mathcal O\left ( dm(tn)^2 \right)$ operations such that, on average, at most $d$ active actuators are selected, and the following bound
\[{\rho(\mathcal W_s(t))} \leq \left ( 1- \sqrt \frac{n}{dt} \right)^{-2} {\rho(\mathcal W(t))}\]
holds for all systemic controllability measures. Moreover, the maximum scaling ratio over all time and inputs is bounded by
\[\max_{i \in [m],k+1 \in [t]} s_i^2(k) ~\leq~\gamma,\]
where $\gamma = \left(1+\sqrt \frac{m}{d}\right)^2$.
\end{corollary}

\begin{proof}
The proof is a simple variation on the proof of Theorem \ref{th-main-approx}, and is not repeated here.
\end{proof}

As expected, the above result shows that the scaling becomes smaller as the ratio $m/d$ decreases. 

 \begin{algorithm}[t]
 {\small
    \SetKwInOut{Input}{Input}
    \SetKwInOut{Output}{Output}

    \Input{$A \in \R^{n \times n}$, $B \in \R^{n \times m}$, $t$ and $d$}
            \vspace{.1cm}
    \Output{$s_i(k) \geq 0$ for $(i,k+1) \in [m] \times [t]$}
    		\vspace{.1cm}
$\mathcal C(t) := \left [ B~AB~A^2B~\cdots~A^{t-1}B \right ]$\\
        \vspace{.1cm}
Set $V= \left(\mathcal C(t) \mathcal C^\top(t)\right)^{-\frac{1}{2}}\mathcal C(t)$\\
        \vspace{.1cm}
Set \[U~=~\frac{1}{\sqrt{t}} \left[ \underbrace{\text e_1, \ldots,\text e_m}_{=I_m}, \cdots, \underbrace{\text e_1, \ldots,\text e_m}_{=I_m} \right]\]\\
\tcp{where $\text e_i \in \R^{m}$ for $i \in [m] $ are the standard basis vectors for $\R^{m}$ and $UU^\top = I_m$}
        \vspace{.1cm}
Run $[c_1, \cdots, c_{mt}] = \DualSet^* (V, U,dt)$\\
        \vspace{.1cm}
\Return $s_i(k):=\sqrt{c_{i+mk}}$ for $(i,k+1) \in [m] \times [t]$\\
        \vspace{.1cm}
\caption{\small A deterministic greedy-based algorithm to construct a sparse weighted actuator schedule (Corollary \ref{th-main-3}).}
    \label{al-max-2}}
\end{algorithm}

{\begin{corollary}
\label{th-main-3}
Given the  time horizon $t \geq n$, model \eqref{model:a}, and $d > 1$,  
Algorithm \ref{al-max-2} deterministically constructs an actuator schedule for \eqref{model:b} in $\mathcal O\left ( dm(tn)^2 \right)$ operations such that it has, on average, at most $d$ active actuators, and the following
\[{\rho(\mathcal W_s(t))} \leq \left ( 1- \sqrt \frac{n}{dt} \right)^{-2} {\rho(\mathcal W(t))}\]
holds for all systemic controllability measures. Moreover, the sum of scaling ratios for all inputs is bounded by
\[\max_{i \in [m]} \sum_{k=0}^{t-1} s^2_i(k) ~\leq~  \gamma,\]
where $\gamma = t \left(1+\sqrt{\frac{m}{dt}}\right)^2$.

\end{corollary}}
\begin{proof}
{The proof is a simple variation on the proof of Theorem \ref{th-main-approx}, and is not repeated here.}
\end{proof}

 \begin{algorithm}[t]
 {\small
    \SetKwInOut{Input}{Input}
    \SetKwInOut{Output}{Output}

    \Input{$A \in \R^{n \times n}$, $B \in \R^{n \times m}$, $t$ and $d$}
            \vspace{.1cm}
    \Output{$s_i(k) \geq 0$ for $(i,k+1) \in [m] \times [t]$}
    		\vspace{.4cm}
$\mathcal C(t) := \left [ B~AB~A^2B~\cdots~A^{t-1}B \right ]$\\
        \vspace{.1cm}
Set $V= \left(\mathcal C(t) \mathcal C^\top(t)\right)^{-\frac{1}{2}}\mathcal C(t)$\\
        \vspace{.1cm}
Set \[U~=~ \frac{1}{\sqrt{m}} \left [\underbrace{\text e_1, \ldots,\text e_1}_{m~\text{times}}, \cdots, \underbrace{\text e_t, \ldots,\text e_t}_{m~\text{times}} \right]\]\\
\tcp{where $\text e_i \in \R^{t}$ for $i \in [t]$ are the standard basis vectors for $\R^{t}$ and $UU^\top = I_t$}
        \vspace{.1cm}
Run $[c_1, \cdots, c_{mt}] = \DualSet^* (V, U,dt)$\\
        \vspace{.1cm}
\Return $s_i(k):=\sqrt{c_{i+mk}}$ for $(i,k+1) \in [m] \times [t]$\\
        \vspace{.1cm}
\caption{\small A deterministic greedy-based algorithm to construct a sparse weighted actuator schedule (Corollary \ref{th-main-4}).}
    \label{al-max-3}}
\end{algorithm}

{\begin{corollary}
\label{th-main-4}
Given the  time horizon $t \geq n$, model \eqref{model:a}, and $d > 1$,  
Algorithm \ref{al-max-3} deterministically constructs an actuator schedule for \eqref{model:b} in $\mathcal O\left (dm(tn)^2\right)$ operations such that it has, on average, at most $d$ active actuators, and the following
\[{\rho(\mathcal W_s(t))} \leq \left ( 1- \sqrt \frac{n}{dt} \right)^{-2} {\rho(\mathcal W(t))}\]
holds for all systemic controllability measures. Moreover, the sum of scaling ratios at each time is bounded by
\[\max_{k+1 \in [t]} \sum_{i=1}^{m} s^2_i(k) ~\leq~\gamma,\]
where $\gamma =  m\left(1+\sqrt{\frac{1}{d}}\right)^2$.
\end{corollary}}
\begin{proof}
The proof is a simple variation on the proof of Theorem \ref{th-main-approx}, and is not repeated here.
\end{proof}

We use a different idea in Subsection \ref{sec:leverage}, to develop scalable algorithms that sparsify control inputs by employing a sub-sampling
method for a time-varying actuator schedule. This however come at the cost of an extra $\log$ factor in terms of the average number of selected actuators.

\subsection{Randomized Approach: Sampling Based on the Leverage Score}
\label{sec:leverage}
In this subsection, we focus on a computationally tractable method for the weighted sparse actuator scheduling problem that achieve near optimal solution.

\begin{definition}
The leverage score of the $i$-th column of matrix $P \in \R^{n \times m}$ is defined as
\[ \ell_i~=~ p_i^\top (PP^\top)^\dag p_i, \]
where $p_i$ is the $i$-th column of matrix $P$.
\end{definition}

This quantity encodes the importance of the $i$-th column compared to the other columns. A larger leverage score shows that the corresponding column has more influence on the spectrum of $P$. Based on the leverage score definition, we get $\ell_i \in [0,1]$ for all $i \in [m]$. Because $\ell_i$'s are the diagonal elements of the projection matrix $P^\top (PP^\top)^{-1}P$ and the diagonal elements of the projection matrix are between zero and one. Leverage score $\ell_i=1$ means that the $i$-th column has a component orthogonal to the rest of the columns. Therefore, eliminating that column will decrease the rank of matrix $P$. On the other hand, $\ell_i=0$ means that the $i$-th column is parallel to the rest of the columns.
 When the corresponding matrix is the graph Laplacian, this quantity reduces to the effective resistance of each link in a graph \cite{Spielman}.

We group the columns of $\mathcal C(t)$ in the following form
\[  \mathcal C(t) = \left [ \underbrace{\left [b_1~Ab_1~\cdots A^{t-1}b_1\right ]}_{\mathcal C_1(t)}~ \cdots ~\underbrace{\left [b_m~Ab_m~\cdots A^{t-1}b_m \right ] }_{\mathcal C_m(t)}\right ], \]
where $b_j$ is the $j$-th column of matrix $B$. Matrix $\mathcal C_j(t)$ presents the controllability matrix of input $j$ at time $t$.  
The leverage score for each column of $\mathcal C(t)$ is defined as
 \begin{equation}
 \ell(A^ib_j) ~=~  (A^{i} b_j)^\top \left (\mathcal C(t) \, \mathcal C^\top(t)\right)^{\dag} A^{i} b_j,
 \label{eq:561}
 \end{equation}
  where $(i+1) \in [t] $ and $j \in [m]$. For these  scores, we have
\begin{eqnarray} \sum_{i=0}^{t-1}\sum_{j=1}^m \ell(A^ib_j) &=&\tr \left( \mathcal C^\top(t) (\mathcal C(t) \mathcal C^\top(t))^{\dag} \mathcal C(t) \right)\nonumber \\
 &=&\tr \left(  (\mathcal C(t) \mathcal C^\top(t))^{\dag} \mathcal C(t) \mathcal C^\top(t)\right)\nonumber \\
 &=&\tr (  I_n ) ~=~n.
\label{eq:631}
\end{eqnarray}
In \eqref{eq:631}, we use the fact that $\tr( AB )=\tr(BA)$ (i.e., the matrices in a trace of a product can be switched without changing the result as long as $A$ and $B^\top$ have the same dimensions), and $\Rank(\mathcal C(t)) = n$ (i.e., the system is controllable).

 \begin{algorithm}[t]
 {\small
    \SetKwInOut{Input}{Input}
    \SetKwInOut{Output}{Output}

    \Input{$A \in \R^{n \times n}$, $B \in \R^{n \times m}$, $t$ and $d$}
            \vspace{.1cm}
    \Output{$\{\sigma_i\}_{i=0}^{t-1}$ and $s_i(k-1)$ for $(i,k) \in [m] \times [t]$}
    		\vspace{.4cm}
$\mathcal C(t) := \left [ B~AB~A^2B~\cdots~A^{t-1}B \right ]$\\
\vspace{.2cm}
{\bf set} $\{\sigma_i\}_{i=0}^{t-1}$ to be the empty sets (\ie ~$\sigma_i := \{\}$) \\
\vspace{.2cm}
{\bf set} $s_i(k-1)=0$ for $(i,k) \in [m] \times [t]$ \\
\vspace{.2cm}
{\bf set} {$\pi(i,k)= \frac{\tr \left( (\mathcal C(t) \mathcal C^\top(t))^{\dag} A^{t-k} b_i (A^{t-k} b_i)^\top \right)}{n}$ for all $(i,k) \in [m] \times[t]$}\\
\vspace{.2cm}
        \For{$j = 1$ {\it to} $M:=\lceil {dt} \rceil$}
       { $(i,k)$ $\leftarrow$ sample $(i,k)$ from $[m] \times [t]$ with probability distribution $\pi$ \\
       $\sigma_{k-1} = \sigma_{k-1} \cup \{i\}$\\
       \vspace{.1cm}
       $s_i^2(k-1) = s_i^2(k-1) + \frac{1}{M \pi(i,k)}$}
        \Return $\{\sigma_i\}_{i=0}^{t-1}$ {\text and} $s_i(k-1)$ for $(i,k) \in [m] \times [t]$\\
        \vspace{.1cm}
    \caption{\small A simple randomized algorithm to compute a sparse weighted actuator schedule $\{\sigma_i\}_{i=0}^{t-1}$ (Theorem \ref{th:rand}).}
    \label{alg:rand}}
\end{algorithm}

We now randomly sample the actuators with probabilities proportional to their leverage scores to sparsify control inputs. This sampling occurs across time and over all possible actuators at each time (see Algorithm \ref{alg:rand}). At every time,  each actuator is kept active or inactive  according to  probability $\ell(A^ib_j)/n$ where $(i+1) \in [t] $ and $j \in [m]$. Using \cite[Thm. 1]{Spielman}, we can construct a sampling strategy that utilizes the leverage score to probabilistically choose actuators. The catch is that there is an extra $\log n$ factor in the average number of selected actuators, and potentially  different actuators are chosen at different times. 

\begin{theorem}
\label{th:rand}
Assume that dynamics \eqref{model:a}, time horizon $t \geq n$, and approximation factor $\epsilon \in \left [1/\sqrt{n},1\right)$ are given. Choose a real number {$d=\frac{9c^2 n \log n}{t\epsilon^2}$, where $c$ is  the constant in Lemma \ref{rudelson}.} Then, Algorithm \ref{alg:rand} produces scheduling \eqref{model:b} which is $\left (\epsilon,d \right )$-approximation of \eqref{model:a} with probability of at least 0.5 for sufficiently large $n$.\footnote{We should note that one can repeat Algorithm \ref{alg:rand}, for example $c = 4$ times to get the desired results with more than $1 - 1/c^4= 0.9375$ probability. Moreover, the probability is improved by increasing the number of iteration $c$. Assume $c$ is a constant and does not depend on $n$, $t$, $d$ or $m$. Then, repeating the algorithm $c$ times does not change the time complexity of the algorithm, and we still have an approximately linear time algorithm. Therefore, by repeating the algorithm and choosing the best result, we can obtain the same error bound with higher probability. }
\end{theorem}

\begin{proof}
{The structure of the proof follows from the proof of \cite[Thm. 4]{Spielman}. Let us start with the following projection matrix
	\begin{equation}
	 \Pi ~=~ \mathcal C(t)^\top \mathcal W^{-1}(t)\, \mathcal C(t),
	 \label{eq::1107}
	 \end{equation}
where $\mathcal C(t)$ is $n$-by-$tm$ controllability matrix \eqref{control-matrix} and matrix $\mathcal W(t)=\mathcal C(t) \, \mathcal C^\top(t)$ is given by \eqref{gramian}. 
%
The $tm$-by-$tm$ projection matrix $\Pi$ has eigenvalue at $0$ with multiplicity  $t \times m-n$  and eigenvalue at $1$ with multiplicity  $n$. Therefore, we get 
\begin{equation}
\tr(\Pi)~=~\Rank(\Pi)~=~n.
\label{1134}
\end{equation}
The set $X$ is obtained based on columns of $\Pi$ as follows
\[X ~=~ \big \{ y_j \in \R^n~:~ y_j = \left (\pi_j\right)^{-1/2} \Pi (., j),~{\text{and}}~ j \in [tm] \big\} , \]
where matrix $\Pi$ is given by \eqref{eq::1107}, vector $\Pi (., j)$ is the $j$-th column of $\Pi$, and $\pi_j$ is probability of selecting vector $y_j$ (\ie, $\pi(y_j)= \pi_j$). The probability distribution $\pi$ over $X$ is defined by 
\begin{equation}
 \pi(y_j) := \pi_j = \frac{\Pi(j,j)}{\Rank(\Pi)}= \frac{\Pi(j,j)}{n},
 \label{eq:1028}
 \end{equation}
 where $\Pi(j,j)$ is the $j$-th diagonal element of matrix $\Pi$ and $j \in [tm]$. Based on \eqref{eq::1107}, each columns of $\Pi$ corresponds to the $i$-th input at time $k-1$ where $(i,k) \in [m] \times[t]$. The mapping comes from the controllability matrix structure \eqref{control-matrix}, and we define it as $\mathfrak m(.): [mt] \rightarrow [m] \times [t]$, where 
 \begin{equation}
 \mathfrak m(j) = (i,k)= \big (j-t\lfloor \frac{j}{t}\rfloor,t-\lfloor \frac{j}{m}\rfloor \big ).
 \label{1127}
 \end{equation}
 This means the $j$-th column of $\Pi$ corresponds to  $(i,k)$, where $k=t-\lfloor \frac{j}{m}\rfloor$ and $i=j-t\lfloor \frac{j}{t}\rfloor$. Thus, in Algorithm \ref{alg:rand}, for notational simplicity we denote $\pi(y_j):=\pi(\mathfrak m(j)) =\pi(i,k)$.
For each element of $X$, we have
\begin{eqnarray}
\|y_j\| &=& \left (\pi_j\right)^{-1/2} \| \Pi (., j) \| = \left ( \frac{n}{\Pi(j,j)} \right)^{1/2} \times (\Pi(j,j))^{1/2}\nonumber \\
 &=&  \sqrt n.
 \label{1128}
\end{eqnarray}
where, we use the fact that
\[ \Pi (., j)^\top \Pi (., j) = \Pi(j,j), \]
because $\Pi$ is an orthogonal projection matrix (\ie, $\Pi \Pi = \Pi$). Then, using \eqref{1134} and \eqref{1128}, we have
\begin{eqnarray}
 \E(yy^\top) &=& \sum_{j=1}^{tm} \pi_j y_j y_j^\top ~=~ \sum_{j=1}^{tm} \pi_j \frac{1}{\pi_j}\Pi (., j) \Pi (., j)^\top \nonumber \\
 &=& \Pi \Pi ~=~ \Pi,
 \label{1158}
\end{eqnarray}
where $y$ is a random variable vector with a countable set of outcomes $X$ occurring with probabilities $\pi$ defined by \eqref{eq:1028}.
%
%
Let $\hat y_1, \ldots, \hat y_M$ be independent samples drawn from $\pi$, then, based on Algorithm \ref{alg:rand}, then we have
\begin{eqnarray}
\Pi \Gamma \Pi&=&\sum_{j=1}^{tm} \Gamma(j,j)\Pi(.,j)\Pi(.,j)^\top \nonumber \\
&=&\sum_{j=1}^{tm} s_i^2(k-1)\Pi(.,j)\Pi(.,j)^\top \nonumber \\
&=& \sum_{j=1}^{tm} \frac{\text{\footnotesize \# of times $(i,k)$ is sampled}}{M \pi(i,k)} \Pi(.,j)\Pi(.,j)^\top \nonumber \\
&=& \frac{1}{M} \sum_{j=1}^{tm} \frac{\text{\footnotesize \# of times $(i,k)$ is sampled}}{ \pi(i,k)} \Pi(.,j)\Pi(.,j)^\top \nonumber \\
&=& \frac{1}{M} \sum_{j=1}^M \hat y_j \hat y_j^\top,
\label{1170}
\end{eqnarray}
where $\Gamma$ is a nonnegative diagonal matrix and the random entry $\Gamma(j, j)$ specifies the ``amount" of the $i$-th input at time $k-1$ (where $\mathfrak m(j)=(i,k)$) included in the sparse actuator scheduling by Algorithm \ref{alg:rand}. For instance, $\Gamma(j, j) = 1/M\pi(\mathfrak m(j))$ if he $i$-th input at time $k-1$ is sampled once, $2/M\pi(\mathfrak m(j))$ if it is sampled twice, and zero if it is not sampled at all. The scaling of the $i$-th input at time $k-1$ in the scheduling is given by $ s_i^2(k-1)= \Gamma (j,j)$ where $\mathfrak m(j)=(i,k)$.
We next use a concentration lemma to prove this theorem. Using Lemma \ref{rudelson}, \eqref{1128}, \eqref{1158}, and \eqref{1170}, we get
\begin{eqnarray}
\E \left \|\frac{1}{M} \sum_{i=1}  \hat y_i \hat y_i^\top-\E y y^\top \right \|_2 &=& \E  \| \Pi \Gamma \Pi - \Pi \|_2 \\
& \leq& \min \left (1, c\sqrt{\frac{n\log M}{M}}\right),
\end{eqnarray}
where $c$ is an absolute constant.
Assuming $M= 9c^2 n \log n/\epsilon^2$ gives\footnote{It can be shown that $M=4 n \log n\epsilon^2$ would be enough to get $\epsilon$ approximation with high probability \cite{LectureNote}.}
\begin{eqnarray}
 \E  \| \Pi \Gamma \Pi - \Pi \|_2 &\leq& c\sqrt{\frac{n\log M}{M}} \nonumber \\
 &\leq&  \epsilon \sqrt{\frac{\log(9c^2 n \log n/\epsilon^2)}{9 \log n}}  \leq \epsilon/2,
 \label{1186}
\end{eqnarray}
for $n$ sufficiently large, and $\epsilon$ is assumed to be in $\left [1/\sqrt{n},1\right)$.
By Markov's inequality and \eqref{1186}, we have
\[ \mathbf{Pr} \left[ \left \| \Pi \Gamma \Pi - \Pi \right \| > \epsilon \right]~\leq~0.5,\]
}
which means we have 
\begin{equation}
 \| \Pi -\Pi \Gamma \Pi \|_2~\leq~ \epsilon,
 \label{eq:110}
 \end{equation}
with probability of at least $0.5$. Note that $\Gamma$ is a non-negative diagonal matrix with weights $s_i^2(k)$ on its diagonal such that $\mathcal W_s(t) = \mathcal C(t) \, \Gamma \, \mathcal C^\top(t)$. Based on \cite[Lemma 4]{Spielman}, the inequality \eqref{eq:110} is equivalent to
 \begin{equation}
\sup_{x \in \R^{tm} \atop x \neq 0 } \frac{| x^\top (\Pi -\Pi \Gamma \Pi) x |}{x^\top x}~ \leq~ \epsilon.
\label{eq:1125}
 \end{equation}
Since we have $\Image\{\mathcal C^{\top}(t) \} \subset \R^{mt}$, it follows that 
 \begin{eqnarray*}
 \sup_{x\in \Image\{  \mathcal C^{\top}(t)\} \atop x \neq 0 } \frac{| x^\top (\Pi -\Pi \Gamma \Pi) x |}{x^\top x} & \leq & \sup_{x \in \R^{mt} \atop x \neq 0 } \frac{| x^\top (\Pi -\Pi \Gamma \Pi) x |}{x^\top x} \\
& \leq & \epsilon.
\end{eqnarray*}
Let us define $x= \mathcal C^\top(t) x'$. Then, we rewrite \eqref{eq:1125} as follows
 \begin{equation}
\sup_{x' \in \R^n \atop x'\notin \Ker\{ \mathcal C^\top(t)\}} \frac{| x'^\top (\mathcal W(t) - \mathcal W_s(t)) x' |}{x'^\top \mathcal W(t) x'} ~\leq~ \epsilon.
\label{eq:?}
 \end{equation} 
As a result, it follows that 
 \begin{equation}
\sup_{x' \in \R^n \atop x'\neq 0} \frac{| x'^\top (\mathcal W(t) - \mathcal W_s(t)) x' |}{x'^\top \mathcal W(t) x'} ~\leq~ \epsilon,
 \end{equation}
 which implies that 
 \begin{equation}
 (1-\epsilon) \mathcal W(t)  ~\preceq~ \mathcal W_s(t) = \mathcal C(t) \Gamma  \mathcal C^\top(t)  ~\preceq~ (1+\epsilon) \mathcal W(t).
 \label{eq:695}
 \end{equation}
 Finally, using \eqref{eq:695} and Definition \ref{def:aprox}, {it is straightforward to show that for every systemic controllability measure $\rho: \mathbb S_+^n \rightarrow \R_+$, we have
\begin{equation*}
	\left| \frac{{\rho}(\mathcal W(t))-{\rho}(\mathcal W_s(t))}{{\rho}(\mathcal W(t))}	\right|~\leq~ \epsilon.	
	\end{equation*}  }
Therefore, we conclude the desired result.
\end{proof}

This result shows that with a simple randomized sampling strategy, one can choose on average less than $\mathcal O(\log n/\epsilon^2)$ number of actuators at each time, to approximate any of  the controllability metrics when $t=n$. {Moreover, this result shows that it is possible to have a time-varying actuator schedule with a constant number of active actuators on average over a time horizon a little longer than $n$ (i.e., $t=\mathcal O(n \log n)$) via random sampling. {Algorithm \ref{alg:rand} computes the sparse actuator schedule using a  nearly-linear time $\tilde{\mathcal O}(mt)$ algorithm\footnote{$ f(n)\in {\tilde {\mathcal O}}\left(g(n)\right)$ means that there exists $c>0$ such that $f(n)\in \mathcal O\left (g(n)\log ^{c}g(n)\right )$.} with guaranteed performance bounds, where $mt$ is the total number of actuations (time-to-control $\times$ number of inputs ). This favorable almost-linear-time complexity is achieved by random sampling of actuators in both time and domain based on their leverage scores \cite{Spielman}.}
According to Theorem \ref{th-main-approx}, the average number of active actuators can be reduced to $\mathcal O(1/\epsilon^2)$, at the expense of either solving SDPs~\cite{lee2017sdp} or greedily handling certain eigenvalue bounds (see Algorithm \ref{al-approx}). Algorithm \ref{alg:rand} is conceptually simpler than Algorithm \ref{al-approx} and the SDP-based algorithm presented in \cite{lee2017sdp}, which provide $d=\mathcal O(1/\epsilon^2)$ in $\mathcal O(m(tn)^2/\epsilon^2)$ and $\tilde{ \mathcal O} \left( mt / \epsilon^{\mathcal O(1)}\right)$  time, respectively.

The concept of a  leverage score for each column can be generalized to a group of columns as follows
\begin{equation}
 \ell_{\mathcal C_i} ~=~ \tr \big( \mathcal C_i^\top(t) \left (\mathcal C(t) \, \mathcal C^\top(t)\right)^{\dag} \mathcal C_i(t) \big).
 \end{equation} 
 Using group leverage scores, one can also use a greedy heuristic algorithm to obtain an approximation solution for the static scheduling problem. We note that the problem of approximation of the controllability Gramian with a sparse, static actuator set is considerably more challenging as it doesn't lend itself to a sampling-based strategy: any choice made at one time has to be consistent with the next. 
 
When using a time-varying schedule, the contribution of each actuator to the Gramian  at each time  is  a rank-one matrix. Therefore, we can use the machinery developed for the Kadison-Singer conjecture to find a sparse subset  of actuators over time to approximate the (potentially very large) sum of rank-one matrices. In the static case, however, the choices of actuators at different times are all the same. As a result, the Gramian can be written as a sum of positive semi-definite  matrices corresponding to the selected actuators at each time.  Finding a sparse approximation in this case would require a generalization of the Kadison-Singer conjecture from sums of rank-one to sums of higher ranked positive semidefinite matrices. Such a result has remained elusive as of yet.}

\section{An Unweighted  Sparse Actuator Schedule}
\label{sec:unweighted}
In the previous section, we allowed for re-scaling of the input to come up with a sparse approximation of the Gramian. Here, we assume that the actuator/signal strength cannot be arbitrarily set for individual active actuators and only can be $0$ or $1$. Given a time horizon $t \geq n$, our problem is to compute an actuator schedule $\sigma = \{\sigma_k\}_{k=0}^{t-1}$ where $\sigma_k \subset [m]$ for the system \eqref{model:a}, i.e.,  
\begin{equation}
 x(k+1) ~=~A\,x(k) ~+~\sum_{i \in \sigma_k} b_i \, u_i(k), ~k \in \Z_{+}.
 \label{model:schedul}
 \end{equation}
 As before, the controllability Gramian at time $t$ for schedule \eqref{model:schedul} is given by
\begin{equation}
\mathcal W_{\sigma}(t) ~:=~ \sum_{i=0}^{t-1} \sum_{j \in \sigma_i} (A^{t-i-1} b_j)(A^{t-i-1}b_j)^\top.
\label{W-sigma}
\end{equation}
%

Optimal actuator selection can now be formulated as a combinatorial optimization problem. We consider both static and dynamic actuator schedules, corresponding to time-invariant and time-varying input matrices.

\subsubsection{The Static Scheduling Problem} In this case, all sets $\sigma_i \subset [m]$ for $i+1 \in [t]$ are identical, which means we keep the same schedule at every point in time  for the whole time horizon $t$:
\begin{equation}
\min_{\sigma \in \mathcal S(m,d_{\max})} \rho \left (\sum_{i=0}^{n-1} \sum_{j \in \sigma} (A^i b_j)(A^i b_j)^\top \right),
\label{p:1}
\end{equation}
where 
\begin{equation}
\mathcal S(m,d_{\max}) := \{ \sigma : \sigma \subset [m], \card(\sigma) \leq d_{\max} \},
\end{equation}
where $d_{\max}$ is a given upper bound on the number of active actuators at each time, and $m$ is the total number of actuators.

\subsubsection{The Time-varying Scheduling Problem}
In this case, the optimal dynamic strategy is given as:  
\begin{equation}
\min_{\{\sigma_i \}_{i=0}^{t-1} \in \mathcal S(m,d_{\max},t)} \rho \left (\sum_{i=0}^{t-1} \sum_{j \in \sigma_i} (A^{t-i-1} b_j)(A^{t-i-1} b_j)^\top \right),
\label{p:2}
\end{equation}
where
\begin{equation}
\mathcal S(m,d_{\max},t) :=\left \{ \{\sigma_i \}_{i=0}^{t-1} : \sigma_i \subset [m], \sum_{i=0}^{t-1} \card
(\sigma_i) \leq t d_{\max} \right \},
\end{equation}
and $d_{\max}$ is a given upper bound on the average number of active actuators at each time, i.e., $d_{\max}~\geq~\sum_{i=0}^{t-1} \card( \sigma_i )/t$, where $t$ is a time horizon, and $m$ is the total number of actuators.

The exact combinatorial optimization  problems \eqref{p:1} and \eqref{p:2} are intractable and NP-hard optimization problems; however, it is straightforward to solve a continuous relaxation of these optimization problems where the cost function $\rho$ is convex. {To find a near-optimal solution of optimization problems \eqref{p:1} and \eqref{p:2}, one can use a variety of standard methods for optimal experimental design (greedy methods, sampling methods,  the classical pipage rounding method combined with SDP).  {Specifically, in the case of submodular systemic controllability measures (e.g., D- and T-optimality), the classical rounding method (e.g., pipage and randomized rounding) combined with SDP relaxation results in computationally fast algorithms with a constant approximation ratio \cite{ravi2016experimental}. These approaches are not applicable to non-submodular systemic measures, such as A-, and E-optimality \cite{olshevsky2017non, allen-zhu17e}.} 

 In the following result, we use a result based on regret minimization of the least eigenvalues of positive semi-definite matrices (cf. \cite{allen-zhu17e}) to obtain  a constant approximation ratio for all systemic controllability metrics.

\begin{theorem}
\label{th-const}
Assume that time horizon $t \geq n$, dynamics \eqref{model:a}, systemic controllability metric $\rho: \mathbb S_+^{n} \rightarrow \R$, and $d_{\max} > 2$ are given. Then there exists a polynomial-time algorithm which computes a schedule $\hat \sigma = \{ \hat \sigma_i\}_{i=0}^{t-1}$ that satisfies 
\[  \rho (\mathcal W_{\hat \sigma}(t) )~\leq~ \gamma\left(\frac{d_{\max}t}{n}\right) . \min_{\{\sigma_i\}_{i=0}^{t-1} \in \mathcal S(m,d_{\max},t)} \rho \left ( \mathcal W_\sigma(t) \right),\]
where $\gamma(d_{\max}t/n)$ is a positive constant depending only on $d_{\max}t/n$.
\end{theorem}

\begin{proof}
The proof is a simple variation on the proof of \cite[thm. 1.1]{allen-zhu17e}, and is not repeated here. 
\end{proof}}

The positive constant $\gamma(.)$ in Theorem \ref{th-const} is defined as follows 
\begin{equation} \gamma(\zeta)~=~\min_{y>\frac{3\zeta}{\zeta-2}} \frac{v(2+\frac{v}{\zeta})}{(1-\frac{2}{\zeta})v-3}, \quad \text{where}~\zeta>2,
\label{eq:cons}
\end{equation}
see the proof of \cite[thm. 1.1]{allen-zhu17e}. For example, for $d_{\max}t/n \in \{4, 10 , 50\}$, using \eqref{eq:cons}, we get
\[\gamma(4)~=~2 \left (5 + \sqrt{21}\right)~\approx~ 19.1652,\]
\[\gamma(10)~=~5 \left (\frac{11 + \sqrt{57}}{16}\right)~\approx~ 5.79682,\]
and
\[\gamma(50)~=~25 \left (\frac{17 + \sqrt{33}}{192}\right)~\approx~ 2.96153.\]

Next, we use the results from Section \ref{sec:weighted}  to obtain an unweighted sparse actuator schedule with guaranteed performance bound.

  \begin{algorithm}[t]
   {\small
    \SetKwInOut{Input}{Input}
    \SetKwInOut{Output}{Output}

    \Input{$A \in \R^{n \times n}$, $B \in \R^{n \times m}$, $t$ and $d_{\max}$}
            \vspace{.1cm}
    \Output{$s_i(k)$ for $(i,k+1) \in [m] \times [t]$}
    		\vspace{.4cm}
$\mathcal C(t) := \left [ B~AB~A^2B~\cdots~A^{t-1}B \right ]$\\
        \vspace{.1cm}
Set $V= \left(\mathcal C(t) \mathcal C^\top(t)\right)^{-\frac{1}{2}}\mathcal C(t)$\\
        \vspace{.1cm}
Set \[U~=~ \left[\underbrace{\text e_1, \ldots, \text e_{mt}}_{=I_{mt}}\right]\] \\
\tcp{where $\text e_i \in \R^{mt}$ for $i \in [mt]$ are the standard basis vectors for $\R^{mt}$}
        \vspace{.1cm}
Run $[c_1, \cdots, c_{mt}] = \DualSet^* (V, U,d_{\max}t)$\\
        \vspace{.1cm}
\Return $s_i(k):=\left \lceil \sqrt{c_{i+mk}}/\left(1+\sqrt \frac{m}{d_{\max}}\right) \right \rceil$ for $(i,k+1) \in [m] \times [t]$\\
        \vspace{.1cm}
\caption{\small A deterministic greedy-based algorithm to construct a sparse unweighted actuator schedule (Corollary \ref{th:unweighted}).}
    \label{al-unweighted}}
\end{algorithm}

\begin{corollary}
\label{th:unweighted}
Assume that time horizon $t \geq n$, dynamics \eqref{model:a}, and $d_{\max} > 1$ are given. Then polynomial-time Algorithm \ref{al-unweighted} deterministically constructs an actuator schedule for \eqref{model:b} with $s_i(k) \in \{0,1\}$ such that it has, on average, at most $d_{\max}$ active actuators, and the following
\[  \rho (\mathcal W_{\sigma}(t) )~\leq~ \left (\frac{1+\sqrt{\frac{m}{d_{\max}}}}{1-\sqrt{\frac{n}{d_{\max}t}}}\right)^{2} \rho (\mathcal W(t)),\]
holds for all systemic controllability measures.
%
\end{corollary}

\vspace{.01cm}
\begin{proof}
{The proof is a simple variation on the proof of Theorem \ref{th-main-approx}, and is not repeated here.}
\end{proof}

In view of this result, one can choose any constant number greater than one as the number of active actuators on average to construct a sparse unweighted actuator schedule in order to approximate controllability measures. This, however, comes at the cost of an extra $\left(1+\sqrt \frac{m}{d_{\max}}\right)^2$ factor in terms of the energy cost compared to the weighted sparse actuator schedule (cf. Corollary \ref{th-main-2}).

 \begin{algorithm}[t]
 {\small
    \SetKwInOut{Input}{Input}
    \SetKwInOut{Output}{Output}

    \Input{$A \in \R^{n \times n}$, $B \in \R^{n \times m}$, $t$ and $d$}
            \vspace{.1cm}
    \Output{$B_s \in \R^{n \times d}$, $\rho(\mathcal W_s)$}
    		\vspace{.4cm}	
	$\mathcal W_s := \mathbf 0_{n \times n}$\\	
	 \vspace{.1cm}
        \For{$k = 1$ {\it to} $d$}
       { $j$ $\leftarrow$ find a column of $B$ that returns the maximum value for \begin{small}\[\rho(\mathcal W_s + \alpha I_n) - \rho \left( \mathcal W_s+\sum_{i=0}^{t-1} A^{i} B(:,j)B(:,j)^{\top} (A^{i})^\top + \alpha I_n \right)\]\end{small}
       \tcp{\small $\alpha>0$ is sufficiently small to avoid singularity}
        \vspace{.1cm}
       $B_s \leftarrow \left [B_s, B(:,j)\right]$\\
        \vspace{.1cm}
       $\mathcal W_s =  \sum_{i=0}^{t-1} A^{i} B_{s}B_{s}^{\top} (A^{i})^\top$\\
       \vspace{.1cm}
        $B(:,j) \leftarrow [~]$
       }
       \vspace{.1cm}
        \Return $B_s$, $\rho(\mathcal W_s)$\\
        \vspace{.1cm}
    \caption{\small A greedy heuristic for given $\rho(.)$ which sequentially picks inputs~ $\GreedyStatic(A,B,t,d)$.} 
    \label{alg:greedy}}
\end{algorithm}

 \begin{algorithm}[t]
 {\small
    \SetKwInOut{Input}{Input}
    \SetKwInOut{Output}{Output}

    \Input{$A \in \R^{n \times n}$, $B \in \R^{n \times m}$, $t$ and $d$}
            \vspace{.1cm}
    \Output{$\rho(\mathcal W_s)$}
    		\vspace{.4cm}
	$\mathcal C := \left [ B~AB~A^2B~\cdots~A^{t-1}B \right ]$\\
	$\mathcal C_s := \mathbf 0_{n \times mt}$\\	
	 \vspace{.1cm}
        \For{$k = 1$ {\it to} $M:=\lceil {dt} \rceil$}
       { $j$ $\leftarrow$ find a column of $C$ that returns the maximum value for \begin{small}\[\rho(\mathcal W_s+ \alpha I_n) - \rho \left( \mathcal W_s+ C(:,j)C(:,j)^\top + \alpha I_n\right)\]\end{small}
        \tcp{\small $\alpha>0$ is sufficiently small to avoid singularity}
        \vspace{.1cm}
       $C_s \leftarrow \left [C_s, C(:,j)\right]$\\
        \vspace{.1cm}
       $\mathcal W_s =  C_s C_s^\top$\\
        \vspace{.1cm}
        $C(:,j) \leftarrow [~]$
       }
       \vspace{.1cm}
        \Return $\rho(\mathcal W_s)$\\
        \vspace{.1cm}
    \caption{\small A greedy heuristic for given $\rho(.)$ which sequentially picks inputs and activation times $\GreedyTimeVarying(A,B,t,d)$. }
    \label{alg:greedyVarying}}
\end{algorithm}

\begin{figure*}[h]
	\centering
	 \psfrag{y}[h][h]{ \footnotesize{control inputs}}  
	 \psfrag{x}[t][t]{ \footnotesize{time}}      
	\includegraphics[trim = 0.5 0.5 0.5 0.5, clip,width=.85 \textwidth]{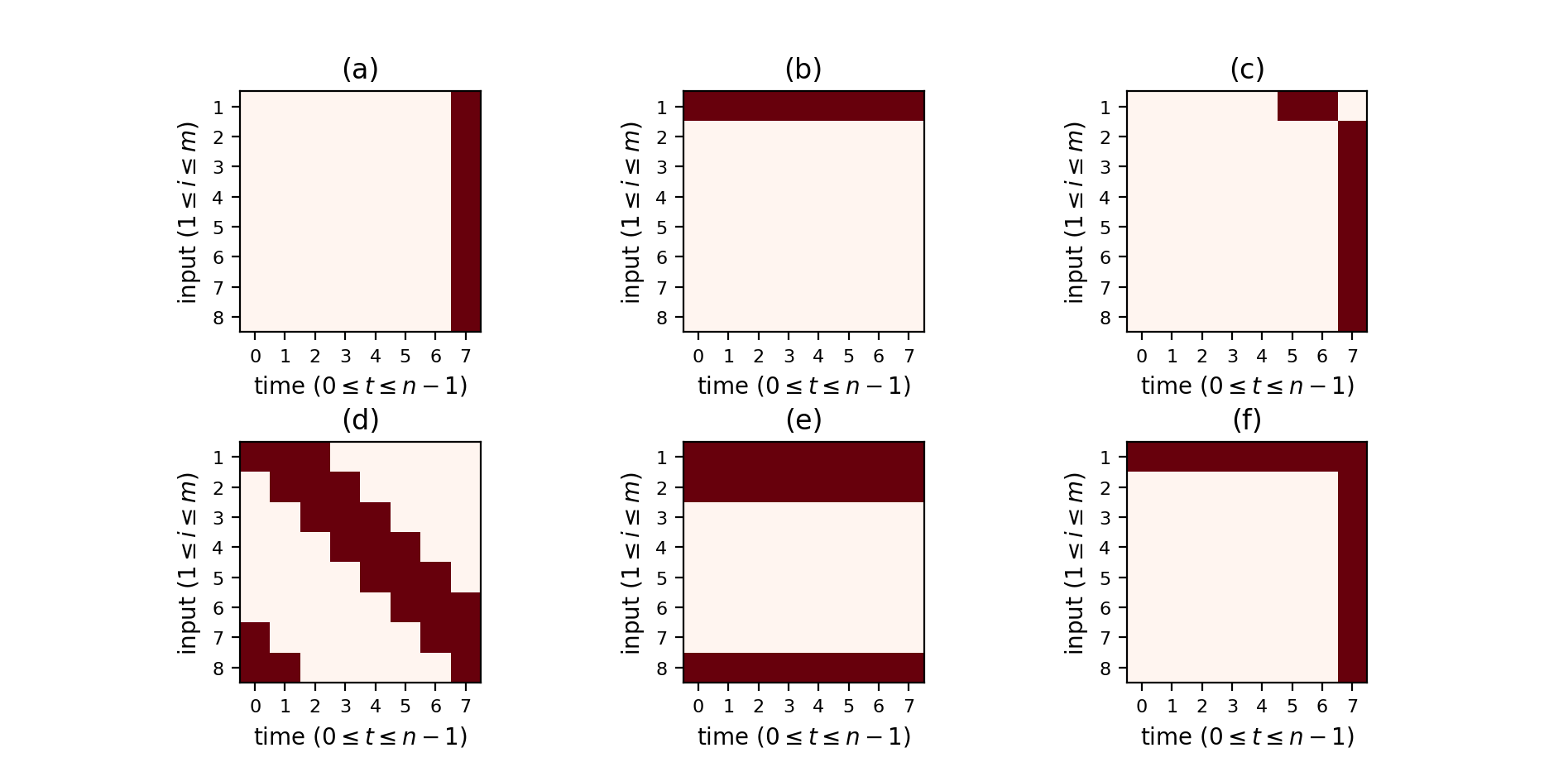} 
	\caption{\small Six unweighted actuator schedules for Example \ref{example_1}: 
      (a) all actuators are active at time $7$ (b) actuator one is active at each time (c) the schedule is obtained Algorithm \ref{al-unweighted} (d) three actuators are active at all time and each actuator is used three times 
      (e) three fixed actuators $\{1,2, 8\}$ are active at all time (f) the proposed sparse schedule based on Algorithm \ref{al-unweighted} with less than two active actuators at each time on average. The color of element $(i,k)$ is red when $s_i(k) =1$ and white otherwise where $i \in [8]$, $k+1 \in [8]$ and $s_i(k) \in \{0,1\}$. For Figs. \ref{fig:6} (c)\&(f), which are obtained based on Algorithm \ref{al-unweighted}, we can observe that the actuator schedule has {procrastination} in actuator activations (i.e., more active actuators at the end of the time horizon); however, in Example \ref{example_3} we can see {``front-loaded"} behavior (i.e., more active actuators early in the time horizon) due to different dynamics in this example.}
	\label{fig:6}
\end{figure*}

\begin{table*}[hh]
	\centering
	{\footnotesize
	\begin{tabular}{ |c|c|c|c|c|c|c|c|c|  }  
	\hline
	&Figs. \ref{fig:6}.(a)\&(b)&Fig. \ref{fig:6}.(c)&Fig. \ref{fig:6}(d)&Fig. \ref{fig:6}.(e)&Fig. \ref{fig:6}.(f)&Algorithm \ref{alg:greedy}&Algorithm \ref{alg:greedyVarying}& Fully Actuated\\
	\hline
	\hline
      $\tr \left (\mathcal W^{-1}(n)\right)$&  uncontrollable & $0.628$&uncontrollable&$0.503$ &$0.161$&  uncontrollable &$0.294$&$0.132$\\ 
      \hline
       $d$& $1$&$1.125$&$3$&$3$&$1.875$&$3$&$3$&$8$\\
      \hline
      \end{tabular}}
      \caption{\small{The values of controllability performance and average number of active actuators at each time for the unweighted actuator schedule presented in Fig. \ref{fig:6} and based on greedy algorithms \ref{alg:greedy} and  \ref{alg:greedyVarying}.  The unweighted schedules presented in Figs. \ref{fig:6} (c)\&(f) are obtained based on Algorithm \ref{al-unweighted}. It is not possible to greedily select three inputs (active at all time) to make the system in Example \ref{example_1} controllable.}}\label{table_2}
\end{table*}

{\section{Numerical Examples}
\label{sec:numerical}
In this section, we consider three numerical examples to demonstrate the results. 

We compare our results with a greedy heuristic that sequentially picks control inputs to maximize the systemic metric decrease of the controllability matrix (see Algorithm \ref{alg:greedy}). The selected inputs are active at all times. It is shown that the greedy method works well and matches the inapproximability barrier\footnote{It approximates the minimum number of inputs in the system that need to be affected for controllability within a factor of $c \log n$ for some $c > 0$.} in polynomial time \cite{Alex2014}.
We also compare our results with a greedy algorithm for a time-varying actuator schedule that sequentially picks both control inputs and activation times to maximize the decrease in the systemic metric of the controllability Gramian (see Algorithm \ref{alg:greedyVarying}).
Without loss of generality, we assume time horizon $t=n$. 

\begin{table*}[hh]
	\centering
	{\footnotesize
	\begin{tabular}{ |c| c| c| c|}  
	\hline
	&~Fig. \ref{fig:sparsity-rand} (Algorithm \ref{alg:rand})~&~Static Leader Schedule~&~Fully Actuated~\\
	\hline
	\hline
       $\tr \left (\mathcal W^{-1}(n)\right)$& 93.64 & 676.68& 18.16  \\ 
      \hline
      Average Number of Leaders:~$d$&~$40$~&~$160$~&~$200$~\\
	\hline
      \end{tabular}}
      \caption{\small{
      The values of controllability performance for three different actuator schedules in Example \ref{example_2}: 1) the weighted actuator schedule in Fig. \ref{fig:sparsity-rand} based on Algorithm \ref{alg:rand},  2) the static leader schedule with $160$  leaders active at all time, 3) the fully actuated case. To have a fair comparison, we normalize the resulting schedule of Algorithm \ref{alg:rand} such that the sum of the scalings satisfies $\sum_{k=0}^{n-1} \sum_{i=1}^m s_i^2(k) = dn$ where $d=40$. The value of the controllability metric for the materialized result of Algorithm \ref{alg:rand} is $18.54$, which is much closer to the controllability metric of the fully actuated case. }}\label{table_example_2}
\end{table*}

\begin{example}[\cite{Alex2014}] \label{example_1}
Assume that the state space matrices of system \eqref{model:a} are given by
\begin{eqnarray}
A~=~\begin{bmatrix}
	1&0&0&0&0&0&0&\frac{-7}{2}\\
 	0&2&0&0&0&0&0&-3\\
 	0&0&3&0&0&0&0&\frac{-5}{2}\\
  	\frac{3}{4}&\frac{1}{2}&0&4&0&0&0&\frac{13}{8}\\
   	0&\frac{3}{4}&\frac{1}{2}&0&5&0&0&\frac{11}{8}\\
    	\frac{5}{4}&0&\frac{3}{4}&0&0&6&0&\frac{3}{2}\\
    	\frac{3}{2}&\frac{5}{4}&1&0&0&0&7&\frac{9}{4}\\
    	0&0&0&0&0&0&0&8
\end{bmatrix}
\label{A_matrix}
\end{eqnarray}
and 
\begin{equation}
B_{\text{min}}~=~\diag \left [1,~ 1,~ 0,~ 0,~ 0,~ 0,~ 0,~ 1\right] 
\label{b_min}
\end{equation}
Direct computation shows that choosing \eqref{b_min} makes the system controllable and no diagonal-matrix sparser than $B_{\text{min}}$ renders $A$ controllable. For this case ($B = B_{\min}$), the performance is:
\[\tr \left ( \sum_{i=0}^{n-1} A^{i} B_{\min}B_{\min}^{\top} (A^{i})^\top \right)^{-1} ~=~ 0.503,\]%
and for the fully actuated case (i.e., $B = I_8$), we have
\[\tr \left ( \sum_{i=0}^{n-1} A^{i} B_{}B{}^{\top} (A^{i})^\top \right)^{-1} ~=~ 0.132.\]
We compare our method with simple-random and periodical switching methods which are depicted in Fig. \ref{fig:6}, and obtain systemic controllability performances, which are presented in Table \ref{table_2}.
\end{example}


\begin{figure}[t]
	\centering
	 \psfrag{y}[h][h]{ \footnotesize{control inputs}}  
	 \psfrag{x}[t][t]{ \footnotesize{time}}      
	\includegraphics[trim = 40.5 30.5 50.5 50.5, clip,width=.4 \textwidth]{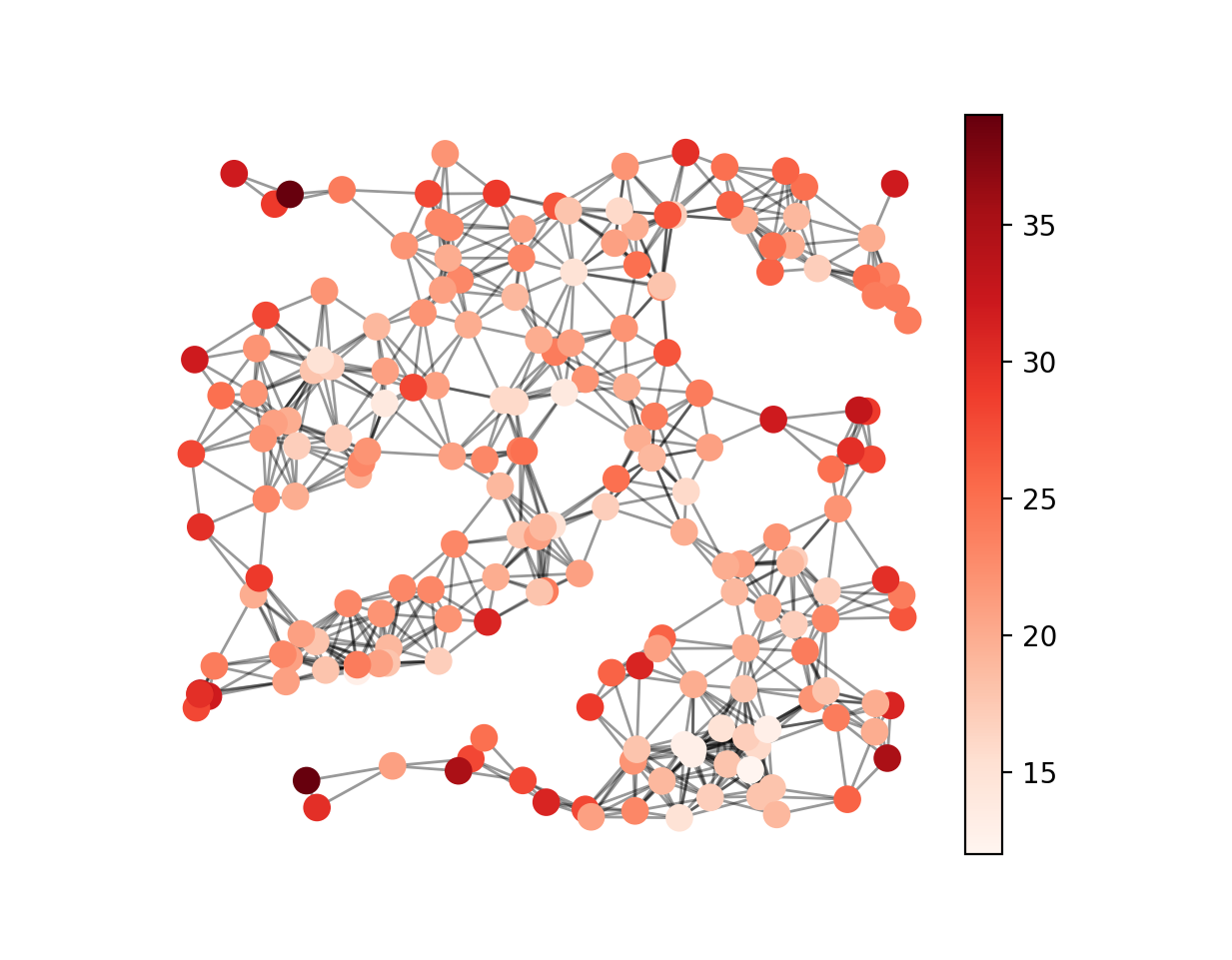} 
	\caption{\small A dynamical network consists of $200$
agents that are randomly distributed in a $1 \times 1$ square-shape
area in space and are coupled over a proximity graph. Every agent
is connected to all of its spatial neighbors within a closed ball of
radius $r = 0.125$. Node colors are proportional to the total number of active steps during time steps $0$ to $199$ from least (white) to greatest (red) based on Algorithm \ref{alg:rand} where $d=40$ (i.e., which means that, on average,  only $20\%$ of agents are controlled at each time).   }
	\label{fig:centrality}
\end{figure}

\begin{figure}[t]
	\centering
	 \psfrag{y}[h][h]{ \footnotesize{control inputs}}  
	 \psfrag{x}[t][t]{ \footnotesize{time}}      
	\includegraphics[trim = 10.5 15.5 10.5 20.5, clip,width=.4 \textwidth]{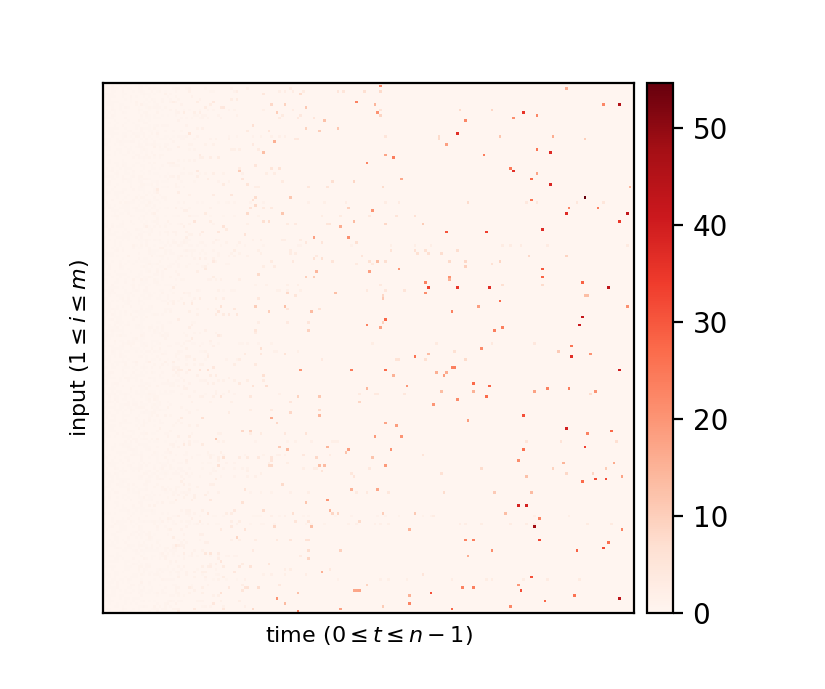} 
	\caption{\small A sparse schedule based on Algorithm \ref{alg:rand} for a given network in Fig. \ref{fig:centrality} where $d=40$. This dynamical network has $m=200$ inputs; however, on average, only $20 \%$ are active at each time between $0$ to $n-1$.  The color of element $(i,k)$ is proportional to the scaling factor $s^2_i(k)$ where $i \in [200]$ and $k+1 \in [200]$.	 }
	\label{fig:sparsity-rand}
\end{figure}

\begin{example} \label{example_2}

Let us consider a dynamic network consisting of $n=200$
agents/nodes, which are randomly distributed in a $1 \times 1$ square-shape
area in space and are coupled over a proximity graph. Every agent
is connected to all of its spatial neighbors within a closed ball of
radius $r = 0.125$. Assume that the state space matrices of this network are given by
\begin{equation}
 A ~=~ I_n - \frac{1}{n} L,~\text{and} ~~ B~=~ I_n,
 \label{matrix::A}
\end{equation}
where $L$ is the Laplacian matrix of the underlying graph given by Fig. \ref{fig:centrality}. Now, we consider the actuator scheduling problem discussed in Section \ref{sec:weighted}.
For undirected consensus networks, a similar problem arises in assignment of a pre-specified number of active agents, as leaders, in order to minimize the controllability metric, e.g., the average controllability energy (cf. \cite{lin2014,rahmani2009}). 
In our setup, each leader $i$ in addition to relative information exchange with its neighbors (based on Laplacian matrix $L$), it also has access to a control input $u_i(.)$.
This system is controllable with only a few inputs/leaders\footnote{The system is not controllable with only one input, because $A$ does not have distinct eigenvalues {\cite{rahmani2009}}.}; however, the amount of the average control energy with a static actuator/leader schedule is too large even for a large number of leaders (see Table \ref{table_example_2}). On the other hand, with a time-varying strategy, the resulting performance is close to the fully actuated case even with a small number of leaders. 
 Therefore, instead of choosing the same leaders at every time step, we choose/switch leaders over a given time horizon to further decrease the controllability metric. 

Fig. \ref{fig:centrality} shows the underlying graph, and node colors are proportional to the total number of active steps during time steps $0$ to $199$ from least (white) to greatest (red). Fig. \ref{fig:sparsity-rand} depicts a sparse schedule based on Algorithm \ref{alg:rand}.

\end{example}

\begin{figure}[t]
	\centering       
	\includegraphics[trim = 0.5 0.5 0.5 0.5, clip,width=.4 \textwidth]{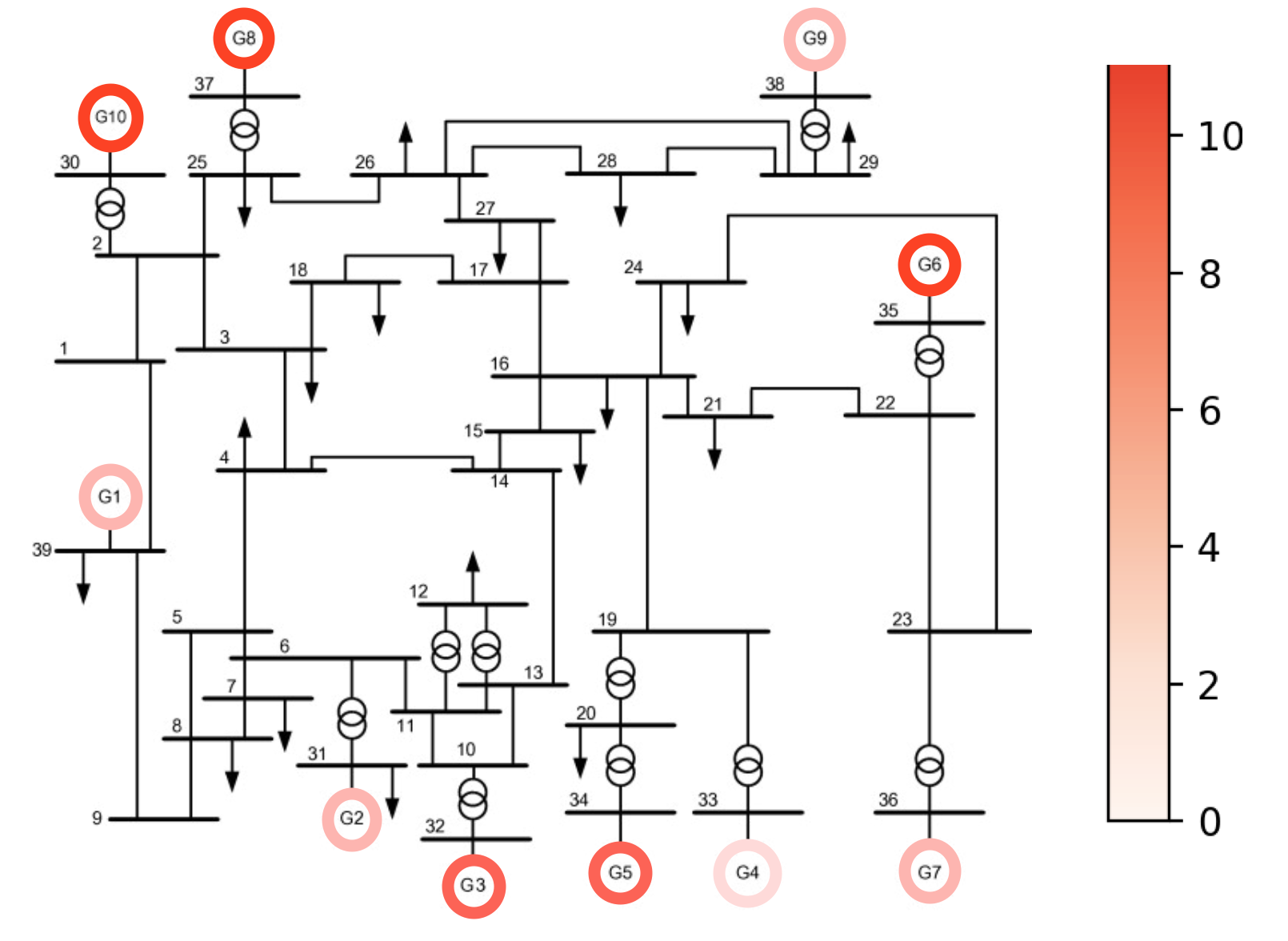} 
	\caption{\small IEEE 10-generator 39-bus power system network (figure is adapted from \cite{atawi2013advance}). Generator colors are proportional to the total number of active steps during time steps $0$ to $19$ from least (white) to greatest (red) based on Algorithm \ref{al-unweighted} where $d=4$ (i.e., which means that, on average,  only four generators are controlled at each time). 	 }
	\label{fig:IEEE-39}
\end{figure}

\begin{figure*}[h]
	\centering
	 \psfrag{y}[h][h]{ \footnotesize{control inputs}}  
	 \psfrag{x}[t][t]{ \footnotesize{time}}      
	\includegraphics[trim = 5.5 5.5 5.5 5.5, clip,width=.8 \textwidth]{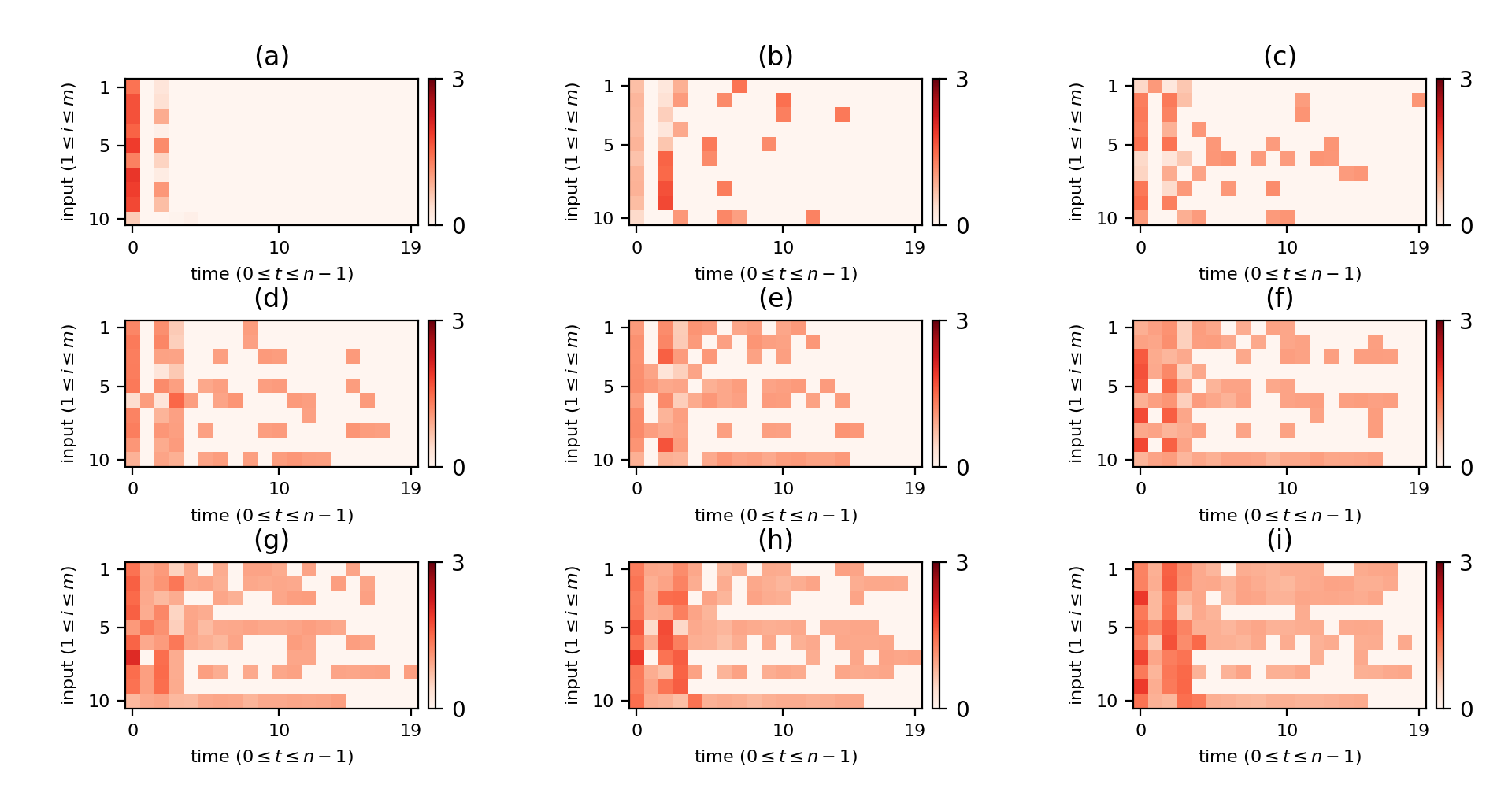} 
	\caption{\small Subplots (a)-(i) presents nine weighted sparse schedules for Example \ref{example_3} based on the proposed deterministic method (Algorithm \ref{al-max-1}) where $d \in \{ 1.05, 1.75,  2.30, 3.10 ,  3.95,  4.60 ,  5.25,  5.75, 6.35\}$ is the average number of active actuators at each time, respectively. The color of element $(i,k)$ is proportional to the scaling factor $s^2_i(k)$ where $i \in [10]$ and $k+1 \in [20]$. }
	\label{fig:sparse_power}
\end{figure*}

\begin{figure}[h]
	\centering
	 \psfrag{y}[h][h]{ \footnotesize{control inputs}}  
	 \psfrag{x}[t][t]{ \footnotesize{time}}      
	\includegraphics[trim = 0.5 0.5 0.5 0.5, clip,width=.4 \textwidth]{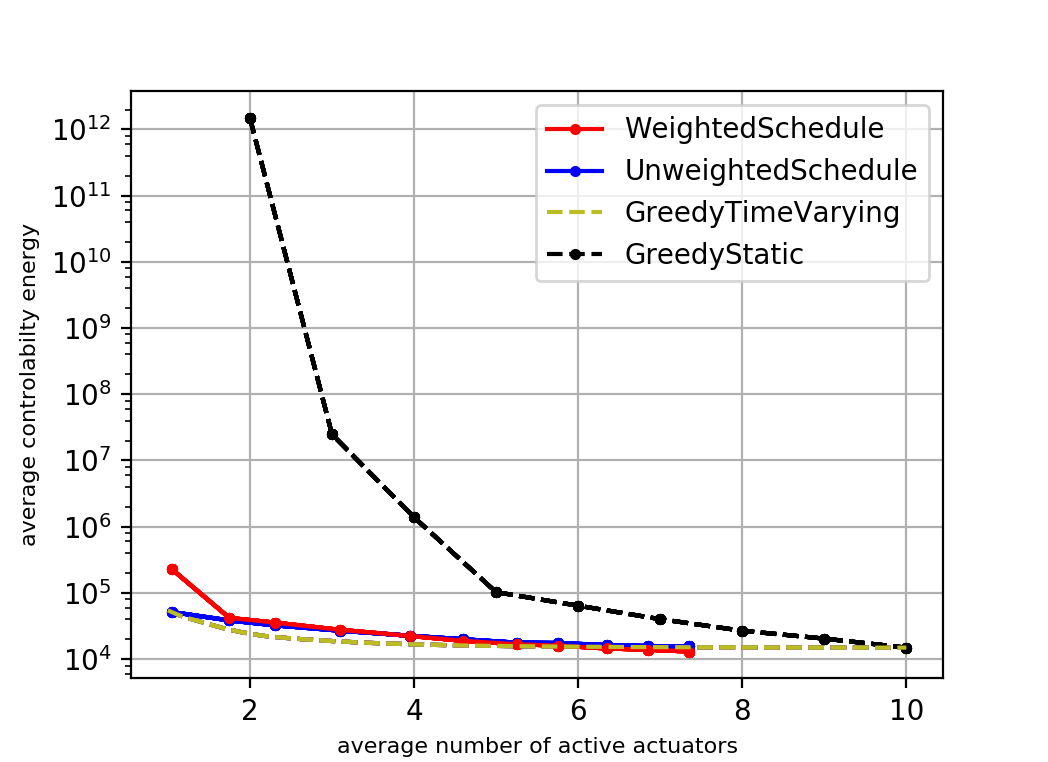} 
	\caption{\small  This plot compares four different methods (Algorithms \ref{al-max-1}, \ref{al-unweighted}, \ref{alg:greedy} and \ref{alg:greedyVarying}) for obtaining sparse actuator schedules of the 10-machine New England Power System in Example \ref{example_3}.	The plot presents the values of average controllability energy (A-optimality) versus the average number of active actuators at each time ($d$). }
	\label{fig:sparse_power_2}
\end{figure}

\begin{example}[Power Network]\label{example_3}

The problem is to select a set of generators to be involved in the wide-area damping control of power systems. We apply our sparse scheduling approach on the IEEE 39-bus test system (a.k.a. the 10-machine New England Power System; see Fig. \ref{fig:IEEE-39}) \cite{atawi2013advance,liu2017minimal}. The single line diagram presented in this figure comprises generators ($G_i$ where $i \in [10]$), loads (arrows), transformers (double circles), buses (bold line segments with number $i \in [39]$), and lines between buses (see \cite{atawi2013advance,liu2017minimal}).

The goal of the wide-area damping control is to damp the fluctuations between generators and synchronize all generators.  The voltage at each generator is adjusted by the control inputs {(e.g., HVDC lines and storages)} to regulate the power output.

We start with a model representing the interconnection between subsystems. Consider the swing dynamics
	\[ m_i \ddot{\theta_i} + d_i \dot{\theta_i} ~=~- \sum_{j \sim i} k_{ij}(\theta_i - \theta_j)+u_i,\]
where $\theta_i$ is the rotor angle state and $w_i := \dot \theta_i$ is the frequency state of generator $i$.
We assume this power grid model consists of $n=10$ generators \cite{liu2017minimal, atawi2013advance}. 
The state space model of the swing equation used for frequency control in power networks can be written as follows 
\begin{eqnarray*}
	\left[ \begin{array}{ccc}
	\dot \theta(t) \\
	\dot w(t) \end{array} \right]  &=& \left[ \begin{array}{ccc}
		0 & I  \\
		-M^{-1}L & -M^{-1} D  \end{array} \right] \left[ \begin{array}{ccc}
		\theta(t) \\
		w(t) \end{array} \right]\nonumber\\
		&&~+~\left[ \begin{array}{ccc}
		0\\
		M^{-1} \end{array} \right] u(t) \label{formation-1}
		\\
 	y(t) &=& \left[ \begin{array}{ccc}
 		\theta(t) \\
 		w(t) \end{array} \right]  	
\end{eqnarray*}
where $M$ and $D$ are diagonal matrices with inertia coefficients and damping coefficients of generators and their diagonals, respectively. 

We assume that both rotor angle and frequency are available for measurement at each generator. This means each subsystem in the power network has a phase measurement unit (PMU). The PMU is a device that measures the electrical waves on an electricity grid using a common time source for synchronization. The system is discretized to the discrete-time LTI system with state matrices $A$, $B$, and $C$ and the sampling time of $0.2$ second (the matrices are borrowed from \cite{ghazal}). 

Fig. \ref{fig:sparse_power} depicts nine sparse schedules based on the proposed deterministic method (Algorithms \ref{al-max-1}) for different values of $d$. The sparsity degree of each schedule is captured by $d$. As $d$ increases the number of non-zero scalings (i.e., activations) increases while the controllability metric decreases (improves).
Fig. \ref{fig:sparse_power_2} compares the results of Algorithms \ref{al-max-1}, \ref{al-unweighted}, \ref{alg:greedy}, and \ref{alg:greedyVarying}.	The plot presents the values of the average controllability energy (A-optimality) versus the average number of active actuators. {To have a fair comparison, we normalize the resulting schedules of all the methods such that the sum of all the scalings satisfies $\sum_{k=0}^{n-1} \sum_{i=1}^m s_i^2(k) = nd$.}

As one expects, Algorithms \ref{al-max-1}, \ref{al-unweighted}, and \ref{alg:greedyVarying} outperform  Algorithm \ref{alg:greedy}.
One observes that Algorithms  \ref{al-max-1}, \ref{al-unweighted} perform nearly as optimal as the time-varying greedy method \ref{alg:greedyVarying}; however, based on our results, we have theoretical guaranteed performance bounds for Algorithms  \ref{al-max-1} and \ref{al-unweighted}.
Furthermore, the usefulness of Algorithms \ref{al-max-1}, \ref{al-unweighted} accentuates itself when the number of active actuators on average is not too small; and potentially can result in a better solution compare to Algorithm \ref{alg:greedyVarying} (see Fig. \ref{fig:sparse_power_2}).

\end{example}

\section{Concluding Remarks} 

In this paper, we have shown how recent advances in matrix reconstruction and graph sparsification literature can be utilized to develop subset selection tools for choosing a relatively  small subset of actuators to approximate certain  controllability measures. Current approaches based on polynomial time relaxations of the subset selection problem require an extra multiplicative factor of $\log n$ sensors/actuators times the minimal number in order to just maintain controllability/observability. Furthermore, when the control energy is chosen as the cost, submodularity-based approaches fail to guarantee the performance using greedy methods.
In contrast, we show that there exists a polynomial-time actuator schedule that on average selects only a constant number of actuators at each time, to approximate controllability measures. Similar results can be developed for the sensor selection problem.
 A potential future direction is to see whether this approach can be used to develop an efficient scheme for minimal reachability problems.

\begin{spacing}{1}
\bibliography{main_Milad}
\end{spacing}

\vspace{-1.cm}

\begin{IEEEbiography}[{\includegraphics[width=1in,height=1.25in,clip,keepaspectratio, trim= 0 0 0 0 ]{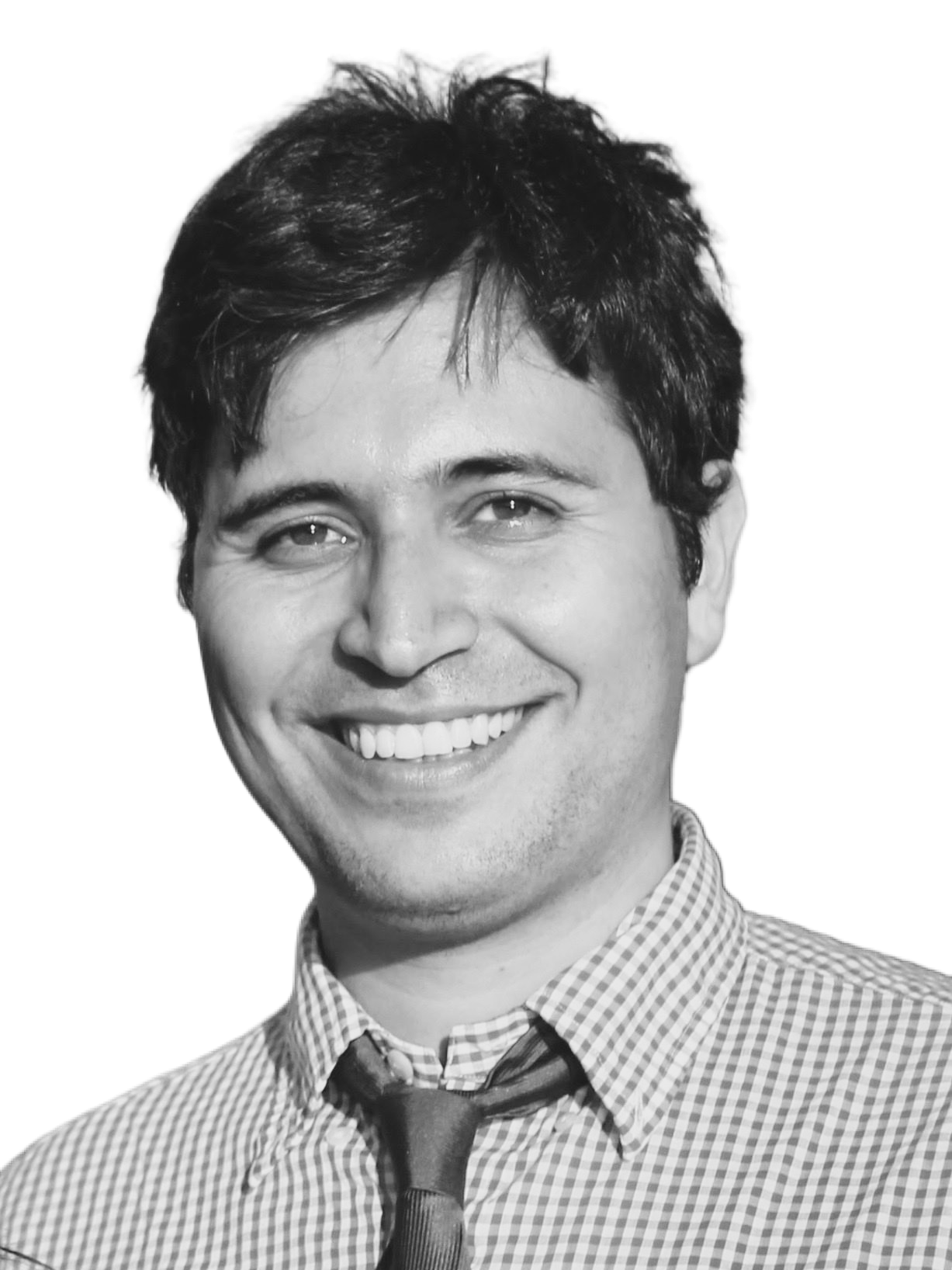}}]{Milad Siami} 
(S'12-M'18) received his dual B.Sc. degrees in electrical engineering and pure mathematics from Sharif University of Technology in 2009, M.Sc. degree in electrical engineering from Sharif University of Technology in 2011. He received his M.Sc. and Ph.D. degrees in mechanical engineering from Lehigh University in 2014 and 2017 respectively. 
From 2009 to 2010, he was a research student at the Department of Mechanical and Environmental Informatics at the Tokyo Institute of Technology, Tokyo, Japan. He was a postdoctoral associate in the Institute for Data, Systems, and Society at MIT, from 2017 to 2019. He is currently an Assistant Professor with the Department of Electrical \& Computer Engineering, Northeastern University, Boston, MA
His research interests include distributed control systems, distributed optimization, and applications of fractional calculus in engineering. 
Dr. Siami received a Gold Medal at National Mathematics Olympiad, Iran (2003) and the Best Student Paper Award at the 5th IFAC Workshop on Distributed Estimation and Control in Networked Systems (2015). Dr. Siami was awarded  RCEAS Fellowship (2012), Byllesby Fellowship (2013), Rossin College Doctoral Fellowship (2015), and Graduate Student Merit Award (2016) at Lehigh University.
\end{IEEEbiography}

\vspace{-1.cm}

\begin{IEEEbiography}[{\includegraphics[width=1in,height=1.25in,clip,keepaspectratio,  trim= .6in 0 .6in  0]{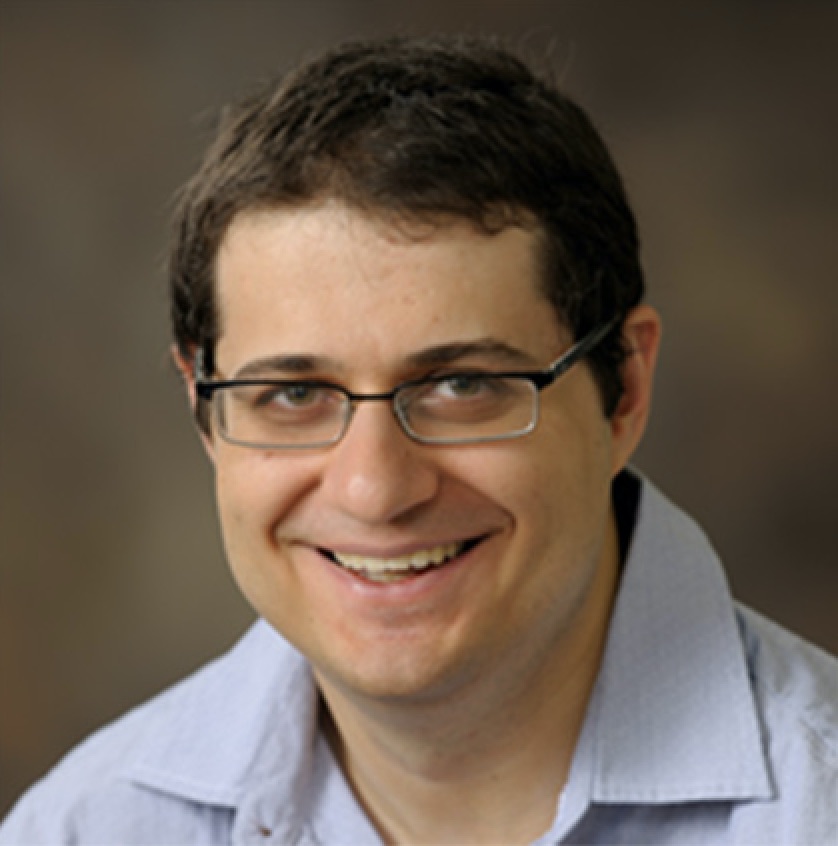}}]{Alex Olshevsky} received the B.S. degree in applied mathematics and in electrical engineering from the Georgia Institute of Technology, Atlanta, GA, USA, both in 2004, and the M.S. and Ph.D. degrees in electrical engineering and computer science from the Massachusetts Institute of Technology, Cambridge, MA, USA, in 2006 and 2010, respectively.
He is currently an Associate Professor at the Department of Electrical and Computer Engineering, Boston University, Boston, MA, USA. Prior to this position, he was a faculty member at the
University of Illinois at Urbana-Champaign. He was a Postdoctoral Scholar in the Department of Mechanical and Aerospace Engineering, Princeton University, from 2010 to 2012 before joining the University of Illinois at Urbana-Champaign in 2012. His research interests are in control theory, optimization, and machine learning, especially in distributed, networked, and multi-agent settings.
Dr. Olshevsky received the NSF CAREER Award, the Air Force Young Investigator Award, the ICS Prize from INFORMS for best paper on the interface of operations research and computer science, and a SIAM Paper Prize for annual paper from the SIAM Journal on Control and Optimization chosen to be reprinted in SIAM Review, and an award in 2019 from the International Medical Informatics Association for best paper on Clinical Research Informatics.  
\end{IEEEbiography}

\vspace{-1.cm}

\begin{IEEEbiography}[{\includegraphics[width=1in,height=1.25in,clip,keepaspectratio,  trim= .2in .1in .2in  .1in]{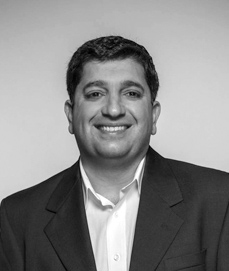}}]{Ali Jadbabaie}
(S'99-M'08-SM'13-F'15) received the B.S. degree from Sharif University of Technology, Tehran, Iran, the M.S. degree in electrical and computer engineering from the University of New Mexico, Albuquerque, NM, USA, and the Ph.D. degree in control and dynamical systems from California Institute of Technology, Pasadena, CA, USA. He is the JR East Professor of engineering, the Associate Director of the Institute for Data, Systems and Society, and the Director of the Sociotechnical Systems Research Center, MIT, Cambridge, MA, USA. He holds faculty appointments in
the Department of Civil and Environmental Engineering and is a Principal Investigator in the Laboratory for Information and Decision Systems. He was a Postdoctoral Scholar at Yale University before joining the faculty at Penn in July 2002. Prior to joining MIT faculty, he was the Alfred Fitler Moore Professor of network science and held secondary appointments in computer and information science and operations, information, and decisions in the Wharton School. His current research interests include the interplay of dynamic systems and networks with specific emphasis on multiagent coordination and control, distributed optimization, network science, and network economics. He is the Inaugural Editor-in-Chief of the IEEE Transactions on Network Science and Engineering and an Associate Editor of the Informs Journal Operations Research. He received the National Science Foundation Career Award, an Office of Naval Research Young Investigator Award, the O. Hugo Schuck Best Paper Award from the American Automatic Control Council, the George S. Axelby Best Paper Award from the IEEE Control Systems Society, and the 2016 Vannevar Bush Fellowship from the office of Secretary of Defense.
\end{IEEEbiography}

\end{document}